\documentclass[conference,compsocconf,onecolumn]{IEEEtran}

\ifCLASSINFOpdf
 
\else
 
\fi
\usepackage{fixltx2e}

\usepackage[margin=1in]{geometry}
\usepackage{comment}
\usepackage[cmex10]{amsmath}
\usepackage{amssymb}
\usepackage{graphicx}
\usepackage{macros}
\usepackage{wrapfig,floatflt,yfonts}
\usepackage{mathtools}
\usepackage{colonequals}
\usepackage{enumerate}
\usepackage{graphicx}
\DeclareGraphicsExtensions{.pdf}
\usepackage{multirow}
\usepackage{algpseudocode}
\usepackage{textcomp}
\usepackage{color}
\usepackage{url}
\usepackage{dsfont}
\usepackage{stmaryrd} 
\usepackage{todonotes}
\usepackage{listings}
\usepackage{colonequals}
\usepackage{xspace}
\usepackage{booktabs}
\usepackage{array}
\usepackage{afterpage}
\usepackage{adjustbox}
\usepackage{tikz}
\usepackage[ruled]{algorithm}
\usepackage{algpseudocode}
\usetikzlibrary{intersections}
\usetikzlibrary{shadows}
\usetikzlibrary{arrows}
\usetikzlibrary{positioning}
\usetikzlibrary{calc}
\graphicspath{ {figures/} }

\tikzset{
	frame/.style={
		rectangle, draw, 
		text width=5em, text centered,
		minimum height=4em,drop shadow,fill=white!40,
		rounded corners,
	},
	line/.style={
		draw, -latex', 
	}
}

\usepackage{algorithmicx}

\newtheorem{definition}{Definition}
\newtheorem{theorem}{Theorem}
\newtheorem{corollary}{Corollary}
\newtheorem{example}{Example}
\newboolean{HIGHCOMM}
\setboolean{HIGHCOMM}{false}
\newcommand \redhighlight[1]{\ifthenelse{\boolean{HIGHCOMM}}{\textcolor{red}{#1}}{#1}}
\newcommand \bhl[1]{\ifthenelse{\boolean{HIGHCOMM}}{\textcolor{blue}{#1}}{#1}}
\newboolean{HIGHLIGHT}
\setboolean{HIGHLIGHT}{false}
\newcommand \yhl[1]{\ifthenelse{\boolean{HIGHCOMM}}{\textcolor{blue}{#1}}{#1}}
\newboolean{ARXIV}
\setboolean{ARXIV}{true}

\begin{document}
\title{Formal Requirement Elicitation and Debugging for Testing and Verification of \\Cyber-Physical Systems}
\author{
\IEEEauthorblockN{Adel Dokhanchi$^{*}$, Bardh Hoxha$^{\dagger}$, and Georgios Fainekos$^{*}$}
$^{*}$School of Computing, Informatics and
Decision Systems \\Arizona State University, Tempe, AZ, U.S.A.\\
Email: \{adokhanc,fainekos\}@asu.edu\\
$^{\dagger}$Department of Computer Science, \\ 
Southern Illinois University, Carbondale, IL, U.S.A. \\
Email: bhoxha@cs.siu.edu
}

\maketitle
\thispagestyle{empty}
\pagestyle{empty}

\begin{abstract}
A framework for the elicitation and debugging of formal specifications for Cyber-Physical Systems is presented. 
The elicitation of specifications is handled through a graphical interface. 
Two debugging algorithms are presented. 
The first checks for erroneous or incomplete temporal logic specifications without considering the system. 
The second can be utilized for the analysis of reactive requirements with respect to system test traces. 
The specification debugging framework is applied on a number of formal specifications collected through a user study. The user study establishes that requirement errors are common and that the debugging framework can resolve many insidious specification errors\footnote{This is the Extended Technical Report of the following ACM-TECS journal paper \cite{DokhanchiHF18} with minor updates in Tables \ref{tab:ATreqs} and \ref{tab:AFreqs}.}.
\end{abstract}

\IEEEpeerreviewmaketitle

\section{Introduction}
Testing and verification of Cyber-Physical Systems (CPS) is important due to the safety critical applications of CPS such as medical devices and transportation systems.
It has been shown that utilizing formal specifications can lead to improved testing and verification \cite{Fainekos2012,JinEtAl13hscc,yang2012querying,KapinskiDJIB15}. 
However, developing formal specifications using logics is a challenging and error prone task even for experts who have formal mathematical training.
Therefore, in practice, system engineers usually define specifications in natural language. 
Natural language is convenient to use in many stages of system development, but its inherent ambiguity, inaccuracy and inconsistency make it unsuitable for use in defining specifications. 

To assist in the elicitation of formal specifications, in \cite{hoxhatowards,Hoxha_ViSpecIROS15}, we presented a graphical formalism and \yhl{the corresponding} tool \textsc{ViSpec} that can be utilized by users in \yhl{both academia and industry}. 
Namely, a user-developed graphical input is translated to a Metric Interval Temporal Logic (MITL) formula. 
The formal specifications in MITL can be used for testing and verification with tools such as \staliro \cite{Annapureddy2011} and Breach \cite{Donze2010}.

In \cite{Hoxha_ViSpecIROS15}, the tool was evaluated through a usability study which showed that \textsc{ViSpec} users were able to use the tool to elicit formal specifications. 
The usability study results also indicated that in a few cases the developed specifications were incorrect. 
This raised two questions. 
First, are these issues artifacts of the graphical user interface? 
Second, can we automatically detect and report issues with the requirements themselves?

We have created an on-line survey\footnote{The on-line anonymous survey is available through:  \url{http://goo.gl/forms/YW0reiDtgi}} to answer the first question. Namely, we conducted a usability study on MITL by targeting users with working knowledge in temporal logics. In our on-line survey, we tested how well formal method users can translate natural requirements to MITL.
That is, given a set of requirements in natural language, users were asked to formalize the requirements in MITL.
The study is ongoing but preliminary results indicate that even users with working knowledge of MITL can make errors in their specifications. 

For example, for the natural language specification \textit{``At some time in the first 30 seconds, the vehicle speed (v) will go over 100 and stay above 100 for 20 seconds"}, the specification $\varphi = \Diamond_{[0,30]} ((v>100) \Rightarrow \Box_{[0,20]} (v>100))$ was provided as an answer by a user with formal logic background.
Here, $\Diamond_{[0,30]}$ stands for ``{\it eventually within 30 time units}" and $\Box_{[0,20]}$ for ``{\it always from 0 to 20 time units}".
However, the specification $\varphi$ is a {\it tautology}!, i.e. it evaluates to true no matter what the system behavior is and, thus, the requirement $\varphi$ is invalid.
This is because, if at some time $t$ between 0 and 30 seconds the predicate $(v > 100)$ is false, then the implication ($\Rightarrow$) will trivially evaluate to true at time $t$ and, thus, $\varphi$ will evaluate to true as well. 
On the other hand, if the predicate $(v > 100)$ is true for all time between 0 and 30 seconds, then the subformula $\Box_{[0,20]} (v>100)$ will be true at all time between 0 and 10 seconds. 
This means that the subformula $(v>100) \Rightarrow \Box_{[0,20]} (v>100)$ is true at all time between 0 and 10 seconds.
Thus, again, $\varphi$ evaluates to true, which means that $\varphi$ is a tautology. 

This implies that specification issues are not necessarily artifacts of the graphical user interface and that they can happen even for users who are familiar with temporal logics.
Hence, specification elicitation can potentially become an issue as formal and semi-formal \yhl{testing and verification methods and tools are being adopted by industry}. 
This is because specification elicitation can be \yhl{performed} by untrained users. 
Therefore, effort can be wasted in checking incorrect requirements, or even worse, the system can pass the incorrect requirements. 
Clearly, this can lead to a false sense of system correctness, which 
leads us to the second question: What can be done \yhl{in an automated way} to prevent specification errors in CPS? 

In this work, we have developed a specification debugging framework to assist in the elicitation of formal requirements. 
The specification debugging algorithm identifies some of the logical issues in the specifications, but not all of them. 
Namely, it performs the following: 

\begin{enumerate}
\item Validity detection: the specification is unsatisfiable or a tautology. 
\item Redundancy detection: the formula has redundant conjuncts. 
\item Vacuity detection: some subformulas do not affect the satisfiability of the formula. 
\end{enumerate}
Redundancy and vacuity issues usually indicate some misunderstanding in the requirements. 
As a result, a wide class of specification errors in the elicitation process can be corrected before any test and verification process is initiated. However, some specification issues cannot be detected unless we consider the system, and test the system behaviors with respect to the specification. 
We provide algorithms to detect specification vacuity with respect to system traces in order to help the CPS developer find more vacuity issues during system testing. Our framework can help developers correct their specifications as well as finding more subtle errors during testing.

This paper is an extended version of the conference paper that appeared in MEMOCODE 2015 \cite{DokhanchiHF15}.
 
\noindent \textbf{Summary of Contributions:}
\begin{enumerate}
\item We present a specification debugging algorithm for a fragment of MITL \cite{AlurFH96} specifications.
\item Using (1) we provide a debugging algorithm for Signal Temporal Logic Specifications \cite{Maler2004}.
\item We extend Linear Temporal Logic (LTL) \cite{mc_gf:CGP99} vacuity detection algorithms \cite{ChocklerS09} to real-time specifications in MITL.
\item We formally define signal vacuity and we provide an algorithm to detect the system traces that vacuously satisfy the real-time specifications in MITL.
\item We present experimental results on specifications that typically appear in requirements for CPS.
\end{enumerate}
The above contributions can help us address and solve some of the logical issues which may be encountered when writing MITL specifications. In particular, we believe that our framework will primarily help users with minimal training in formal requirements who use graphical specification formalisms like \textsc{ViSpec} \cite{Hoxha_ViSpecIROS15}. The users of \textsc{ViSpec} can benefit from our feedback and fix any reported issues. In addition, we can detect potential system-level issues using algorithms to determine specification vacuity with respect to system traces during testing.

In this paper, the new results over the conference version of the paper \cite{DokhanchiHF15} concern item (4) in the list above and are presented in Sections \ref{svc}, \ref{ltlsat}, \ref{afd}, and Appendix \ref{app:SV}.
In addition, we have expanded some examples and added further experimental results which did not appear in \cite{DokhanchiHF15}. 
Furthermore, we added new algorithms in Section 4.

\section{Related works}
\label{related}
The challenge of developing formal specifications has been studied in the past. 
The most relevant works appear in \cite{autili2007graphical} and \cite{zhang2010timed}. 
In \cite{autili2007graphical}, the authors extend Message Sequence Charts and UML 2.0 Interaction Sequence Diagrams to propose a scenario based formalism called Property Sequence Chart (PSC). 
The formalism is mainly developed for specifications on concurrent systems. 
In \cite{zhang2010timed}, PSC is extended to Timed PSC which enables the addition of timing constructs to specifications. 
Another non-graphical approach to the specification elicitation problem utilizes specification patterns \cite{Dwyer1998PSP}. The patterns provided include commonly used abstractions for different logics including LTL, Computation Tree Logic (CTL), or Quantified Regular Expressions (QRE). This work was extended to the real-time domain, \cite{Konrad2005RSP}.

Specification debugging can also be considered in areas such as system synthesis \cite{RamanK11} and software verification \cite{AmmonsMBL03}. 
In system synthesis, realizability is an important factor, which checks whether the system is implementable given the constraints (environment) and requirements (specification) \cite{EhlersR14,KonighoferHB13,CimattiRST08,RamanK11}. 
Specification debugging can also be considered with respect to the environment for robot motion planning. 
In \cite{Fainekos11,KimFS12}, the authors considered the problem where the original specification is unsatisfiable with the given environment and robot actions. 
Then, they relax the specification in order to render it satisfiable in the given domain. 

One of the most powerful verification methods is model checking \cite{mc_gf:CGP99} where a finite state model of the system is evaluated with respect to a specification. 
For example, let us consider model checking with respect to the LTL formula $\varphi=\Box(req\Rightarrow\Diamond ack)$ which represents the following Request-Response requirement ``if at any time in the future a {\it request} happens, then from that moment on an {\it acknowledge} must eventually happen''. 
Here, $\varphi$ can be trivially satisfied in all systems in which a {\it request} never happens.
In other words, if the {\it request} never happens in the model of the system (let's say \yhl{due to a modeling error}), our goal for checking the reaction of the system (issuing the {\it acknowledge}) is not achieved.
Thus, the model satisfies the specification but not in the intended way.
This may hide actual problems in the model.

Such satisfactions are called {\it vacuous} satisfactions.
Antecedent failure was the first problem that raised vacuity as a serious issue in verification \cite{BeattyB94,BenDavidCFR15}. 
Vacuity can be addressed with respect to a model \cite{BeerBER01,KupfermanV03} or without a model \cite{FismanKSV08,ChocklerS09}. 
A formula which has a subformula that does not affect the overall satisfaction of the formula is a vacuous formula.
It has been proven in \cite{FismanKSV08} that a specification $\varphi$ is satisfied vacuously in all systems that satisfy it iff $\varphi$ is equivalent to some mutations of it. 
In \cite{ChocklerS09}, they provide an algorithmic approach to detecting vacuity and redundancy in LTL specifications. 
Vacuity with respect to testing was considered in \cite{BallK08}.
The authors in \cite{BallK08} defined {\it weak vacuity} for test suites that vacuously pass LTL monitors, e.g., \cite{HavelundR04}. 
The main idea behind the work in \cite{BallK08} is that some transitions are removed from the LTL specification automata in order to find vacuous passes during testing.
The authors in \cite{BallK08} renamed the vacuity in model checking as {\it strong vacuity}.
The authors in \cite{PostHP11} consider the problem of vacuity detection in the set of requirements formalized in Duration Calculus \cite{Meyer2008}.

Our work extends \cite{ChocklerS09} and it is applied to a fragment of MITL. 
We provide a new definition of vacuity with respect to Boolean or real-valued signals. To the best of our knowledge, vacuity of real-time properties such as MITL has not been addressed yet.
Although this problem is computationally hard, in practice, the computation problem is manageable, due to the small size of the formulas.
\section{Preliminaries}

In this work, we take a general approach in modeling Cyber-Physical Systems (CPS). 
In the following,  $\Re$ is the set of real numbers, $\Re_+$ is the set of non-negative real numbers, $\Qe$ is the set of rational numbers, $\Qe_+$ is the set of non-negative rational numbers.
Given two sets $A$ and $B$, $B^A$ is the set of all functions from $A$ to $B$, i.e., for any $f \in B^A$ we have $f : A \rightarrow B$. We define $2^A$ to be the power set of set A. Since we primarily deal with bounded time signals, we fix the variable $T \in \Re_+$ to denote the maximum time of a signal.

\subsection{Metric Interval Temporal Logic}

Metric Temporal Logic (MTL) was introduced in \cite{Koymans1990} in order to reason about the quantitative timing properties of boolean signals.
Metric Interval Temporal Logic (MITL) is MTL where the timing constraints are not allowed to be singleton sets \cite{AlurFH96}. 
In the rest of the paper, we restrict our focus to a fragment of MITL called Bounded-MITL($\Diamond$,$\Box$) where the only temporal operators allowed are \textit{Eventually} ($\Diamond$) and \textit{Always} ($\Box$) operators with timing intervals.
Formally, the syntax of Bounded-MITL($\Diamond$,$\Box$) is defined by the following grammar:

\begin{definition}[Bounded-MITL($\Diamond$,$\Box$) syntax]
\label{def:syn}
\end{definition} 
$\phi\;::=\;\top \; | \; \bot \; | \; a \; | \; \neg a \; | \; \phi _1 \wedge \phi_2\; | \; \phi _1 \vee \phi_2 \; | \; \Diamond_I\phi_1 \; | \; \Box_I\phi _1 $\\
\noindent where $a \in AP$, $AP$ is the set of atomic propositions, $\top$ is True, $\bot$ is False. Also, $I$ is a nonsingular interval over $\Qe_+$ with defined end-points. The interval $I$ is right-closed. 
We interpret MITL semantics over timed traces. 
A timed trace is a mapping from the bounded real line to sets of atomic propositions ($\tss : [0,T] \rightarrow 2^{AP}$). We assume that the traces satisfy the finite variability condition (non-Zeno condition)\footnote{The satisfiability tools for MITL that we use in Section \ref{MITLSAT}, \yhl{assume} that the traces satisfy the finite variability condition \cite{BersaniRP16}.}.

\begin{definition}[Bounded-MITL($\Diamond$,$\Box$) semantics in Negation Normal Form (NNF)]
\label{def:mitl}
Given a timed trace $\tss : [0,T] \rightarrow 2^{AP}$  and $t,t' \in [0,T]$, and an MITL formula $\phi$, the satisfaction relation $(\tss,t) \vDash \phi$ is inductively defined as: 
\begin{itemize}
\item[] $(\tss,t) \vDash \top$
\item[] $(\tss,t) \vDash a$ iff $a \in \tss(t)$
\item[] $(\tss,t) \vDash \neg a$ iff $a \not \in \tss(t)$
\item[] $(\tss,t) \vDash  \varphi_1 \wedge \varphi_2$ iff $(\tss,t) \vDash \varphi_1$ and  $(\tss,t) \vDash \varphi_2$
\item[] $(\tss,t) \vDash  \varphi_1 \vee \varphi_2$ iff $(\tss,t) \vDash \varphi_1$ or  $(\tss,t) \vDash \varphi_2$
\item[] $(\tss,t) \vDash  \Diamond_I\varphi_1$ iff $\exists t' \in (t + I) \cap [0,T]$ s.t $(\tss,t') \vDash  \varphi_1$.
\item[] $(\tss,t) \vDash  \Box_I\varphi_1$ iff $\forall t' \in (t + I) \cap [0,T]$,  $(\tss,t') \vDash \varphi_1$.
\end{itemize}
\end{definition} 
Given an interval $I=[l,u]$, $(t + I)$ creates a new interval $I'$ where $I'=[l+t,u+t]$. A timed trace $\tss$ satisfies a Bounded-MITL($\Diamond$,$\Box$) formula $\phi$ (denoted by $\tss \vDash \phi$), iff $(\tss,0) \vDash \phi$. False is defined as $\bot\equiv\neg\top$.
In this paper, we assume that Bounded-MITL($\Diamond$,$\Box$) formula is in Negation Normal Form (NNF)\footnote{We relax this assumption for addressing the Request-Response specifications (see Section \ref{sec:AF}).} where the negation operation is only applied on atomic propositions. NNF is easily obtainable by applying DeMorgan's Law, i.e $\neg\Diamond_I\varphi\equiv\Box_I\neg\varphi$ and $\neg\Box_I\varphi\equiv\Diamond_I\neg\varphi$. 
Any {\it Implication} operation ($\Rightarrow$) will be rewritten as $\psi\Rightarrow\varphi\equiv\psi'\vee\varphi$, where $\psi'\equiv\neg\psi$ and also $\psi'$ is in NNF, and NNF formulas, only contain the following boolean operators of ($\wedge,\vee$).
For simplifying the presentation, when we mention MITL, we mean Bounded-MITL($\Diamond$,$\Box$).
 Given MITL formulas $\varphi$ and $\psi$, $\varphi$ satisfies $\psi$, denoted by $\varphi\models\psi$ iff $\forall \tss.\tss\models\varphi \implies \tss\models\psi$. 
Throughout this paper, we use $\varphi\in\psi$ to denote that $\varphi$ is a {\bf subformula} of $\psi$.
\subsection{Signal Temporal Logic}

The logic and semantics of MITL can be extended to real-valued signals through Signal Temporal Logic (STL) \cite{Maler2004}. 

\begin{definition}[Signal Temporal Logic \cite{Maler2004}]  Let $s : [0,T] \rightarrow \Re^m$ be a real-valued signal, and $\Pi=\{\pi_1,...,\pi_n\}$ be a collection of predicates or boolean functions of the form $\pi_i : \Re^m \rightarrow \mathbb{B}$ where $\mathbb{B}=\{\top,\bot\}$ is a boolean value. 
\end{definition}
For any STL formula $\Phi_{STL}$ over predicates $\Pi$, we can define a corresponding MITL formula $\Phi_{MITL}$ over some atomic propositions $AP$ as follows: 
\begin{enumerate}
\item Define a set $AP$ such that for each $\pi\in \Pi$, there exist some $a_{\pi}\in AP$
\item For each real-valued signal $s$  we define a $\tss$ such that $\forall t.a_{\pi}\in\tss(t)$ iff $\pi(s(t))=\top$
\item $\forall t.(s,t)\vDash\Phi_{STL}$ iff $(\tss,t)\vDash\Phi_{MITL}$
\end{enumerate}
The traces resulting from abstractions through predicates of signals from physical systems satisfy the finite variability assumption. For practical applications, the finite variability assumption is satisfied.
Since our paper focuses on CPS, with the abuse of terminology, we may use signal to refer to both timed traces and signals.
\subsection{Visual Specification Tool}

The Visual Specification Tool (\textsc{ViSpec}) \cite{Hoxha_ViSpecIROS15} enables the development of formal specifications for CPS. 
The graphical formalism enables reasoning on both timing and event sequence occurrence. Consider the specification $\phi_{cps} = \Box_{[0,30]}((speed > 100) \Rightarrow \Box_{[0,40]}(rpm > 4000))$. 
It states that whenever within the first 30 seconds, \textit{vehicle speed} goes over 100, then from that moment on, the \textit{engine speed (rpm)}, for the next 40 seconds, should always be above 4000. 
Here, both the sequence and timing of the events are of critical importance. See Fig. \ref{fig:spec_cps} for the visual representation of $\phi_{cps}$.

Users develop specifications using a visual formalism which can be translated to an MITL formula. 
The set of specifications that can be generated from this graphical formalism is a proper subset of the set of MITL specifications. 
Fig. \ref{fig:grammar} represents the grammar that produces the set of formulas that can be expressed by the \textsc{ViSpec} graphical formalism. 
In Fig. \ref{fig:grammar}, $p$ is an atomic proposition. 
In the tool, the atomic propositions are automatically derived from graphical templates. For example the formula $\Box_\Ic\Diamond_\Ic p$ can be generated using the following parse tree S$\longrightarrow$T$\longrightarrow$C$\longrightarrow\Box_\Ic\Diamond_\Ic$D$\longrightarrow\Box_\Ic\Diamond_\Ic p$.
\textsc{ViSpec} provides a variety of templates and the connections between them, which allow the users to express a wide collection of specifications as presented in Table \ref{tab:specClasses}. For more detailed description of \textsc{ViSpec}, refer to \cite{Hoxha_ViSpecIROS15}.
\begin{figure}
	\noindent
	\begin{minipage}{.50\linewidth}
	\vspace{-10pt} 
\centering
 	\includegraphics[width=6cm]{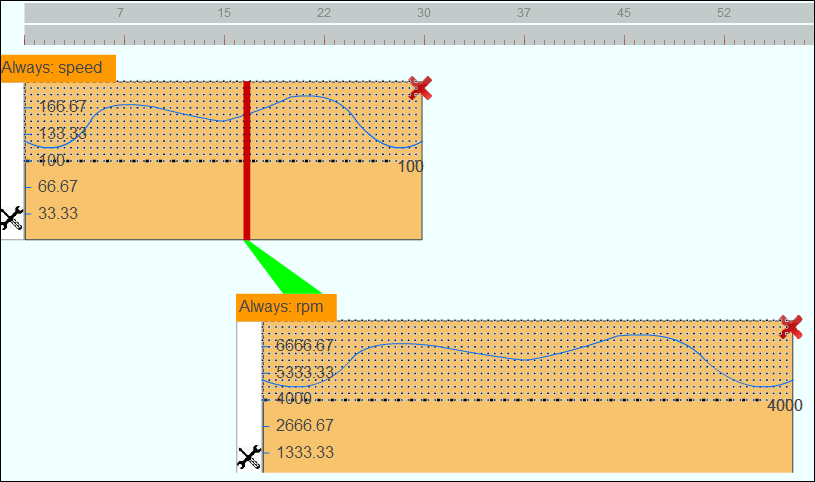}
	 \ifthenelse{\boolean{ARXIV}}{\caption{Graphical representation of  $\phi_{cps}=\Box_{[0,30]}((speed>100)\Rightarrow\Box_{[0,40]}(rpm>4000))$}}{\caption{Graphical representation of \\ $\phi_{cps}=\Box_{[0,30]}((speed>100)\Rightarrow\Box_{[0,40]}(rpm>4000))$}}
	\label{fig:spec_cps}
	\end{minipage}
	\hspace{.03\linewidth}
	\begin{minipage}{.45\linewidth}
		\vspace{-10pt}
	\begin{itemize}
	\item[]
	\item[] S $\longrightarrow$ $\neg$T $|$ T
	\item[] T $\longrightarrow$ A $|$ B $|$ C
	\item[] A $\longrightarrow$ P $|$ (P$\wedge$A) $|$ (P$\Rightarrow$A) 
	\item[] B $\longrightarrow$ $\Box_\Ic$D $|$ $\Diamond_\Ic$D
	\item[] C $\longrightarrow$ $\Box_\Ic\Diamond_\Ic$D $|$ $\Diamond_\Ic\Box_\Ic$D
	\item[] D $\longrightarrow$ $p$ $|$ ($p$$\Rightarrow$A) $|$ ($p$$\wedge$A) $|$ ($p$$\Rightarrow$B) $|$ ($ p $$\wedge$B)
	\item[] P $\longrightarrow$ $p$ $|$ $\Box_\Ic p$ $|$ $\Diamond_\Ic p$
	\end{itemize}
	\caption{\textsc{ViSpec} grammar to generate MITL}
	\label{fig:grammar}	 
	\end{minipage}
\end{figure}
\begin{table*}[t]
	
	\centering
	\ifthenelse{\boolean{ARXIV}}{\caption{Classes of specifications expressible with the graphical formalism\label{tab:specClasses}}}{
		\tbl{Classes of specifications expressible with the graphical formalism\label{tab:specClasses}}}
	{ 
		\begin{tabular}{p{1.6cm}p{11cm}}
			\toprule
			Specification Class & Explanation \\
			\midrule
			Safety & Specifications of the form $\Box \phi$ used to define specifications where $\phi$ should always be true. \\
			Reachability & Specifications of the form $\Diamond \phi$ used to define specifications where $\phi$ should be true at least once in the future (or now). \\
			Stabilization & Specifications of the form $\Diamond \Box \phi$ used to define specifications that, at least once, $\phi$ should be true and from that point on, stay true. \\
			Oscillation & Specifications of the form $\Box \Diamond \phi$ used to define specifications that, it is always the case, that at some point in the future, $\phi$ repeatedly will become true.  \\
			Implication & Specifications of the form $\phi \Rightarrow \psi$ requires that $\psi$ should hold when $\phi$ is true. \\
			Request-Response & Specifications of the form $\Box(\phi \Rightarrow M\psi)$, where $M$ is temporal operator, used to define an implicative response between two specifications where the timing of $M$ is relative to timing of $\Box$. \\
			Conjunction & Specifications of the form $\phi \wedge \psi$ used to define the  conjunction of two sub-specifications. \\
			Non-strict Sequencing & Specifications of the form $N(\phi \wedge M\psi)$, where $N$ and $M$ are temporal operators, used to define a conjunction between two specifications where the timing of $M$ is relative to timing of $N$. \\
			\bottomrule
		\end{tabular}
	}
\end{table*}

\section{MITL Elicitation Framework}
\label{stl2mitl}
Our framework for elicitation of MITL specifications is presented in Fig. \ref{fig:framework}. Once a specification is developed using \textsc{ViSpec}, it is translated to STL. Then, we create the corresponding MITL formula from STL. Next, the MITL specification is analyzed by the debugging algorithm which returns an alert to the user if the specification has inconsistency or correctness issues. The debugging process is explained in detail in the next section. 

\begin{figure}[t]
	\centering
		\scalebox{0.8}{
			\begin{tikzpicture}[font=\small\sffamily,very thick,node distance = 2.67cm]
			\node [frame] (ViSpec) {\textsc{ViSpec} Tool};
			\node [left of = ViSpec] (userInput) {User Input};
			\node [frame,right of = ViSpec] (mitl) {MITL};
			\node [frame,right of = mitl] (AutoDeb) {Debugging};
			\node [right of = AutoDeb] (Specification) {Specification};
			\node [below of = mitl, above = 0.1cm] (ret) {Revision Necessary};
			\path [line] (userInput.0) |- (ViSpec);
			\path [line] (ViSpec) |- (mitl);
			\path [line] (mitl) |- (AutoDeb);
			\path [line] (AutoDeb) |- (Specification);
			\path [line] (AutoDeb) |- ($(mitl.south west) + (-1.65,-1)$) -- (ViSpec.270); 
			\end{tikzpicture}
		}
	\vspace{-8pt}
	\caption{Specification Elicitation Framework}
	\label{fig:framework}
	\vspace{-7pt}
\end{figure}

To enable the debugging of specifications, we must first project the STL predicate expressions (functions) into atomic propositions with independent truth valuations. 
This is very important because the atomic propositions ($a\in AP$) in MITL are assumed to be independent of each other. 
However, when we project predicates to the atomic propositions, the dependency between the predicates restricts the possible combinations of truth valuations of the atomic propositions. 
This notion of predicate dependency is illustrated using the following example.
Consider the real-valued signal $Speed$ in Fig. \ref{fig:stl2mitl}. The boolean abstraction $a$ (resp. $b$) over the $Speed$ signal is true when the $Speed$ is above 100 (resp. 80). The predicates $a$ and $b$ are related to each other because it is always the case that if $Speed>100$ then also $Speed>80$. 
In Fig. \ref{fig:stl2mitl}, the boolean signals for predicates $a$ and $b$ are represented in black solid and dotted lines, respectively. 
It can be seen that solid and dotted lines are overlapping which shows the dependency between them.
However, this dependency is not captured if we naively substitute each predicate with a unique atomic proposition. 
If we lose information about the intrinsic logical dependency between $a$ and $b$, then the debugging algorithm will not find possible specification issues. 

For analysis of STL formulas within our MITL debugging process, we must replace the original predicate with non-overlapping (mutual exclusive) predicates. For the example illustrated in Fig. \ref{fig:stl2mitl}, we create a new atomic proposition $c$ which corresponds to $100 \geq speed > 80$ and the corresponding boolean signal is represented in gray. In addition, we replace the atomic proposition $b$ with the propositional formula $a\vee c$ since $speed > 80\equiv (speed > 100 \vee 100 \geq speed > 80 )$. 
Now, the dependency between $Speed>100$ and $Speed>80$ can be preserved because it is always the case that if $a$ ($Speed>100$), then $a\vee c$ ($Speed>80$). 
It can be seen in Fig. \ref{fig:stl2mitl} that the signal $b$ (dotted line) is the disjunction of the solid black ($a$) and gray ($c$) signals, where $a$ and $c$ cannot be simultaneously true. 

The projection of STL to MITL with independent atomic propositions is conducted using a brute-force approach that runs through all the combinations of predicate expressions to find overlapping parts. 
The high level overview of Algorithm \ref{alg:GMEP} is as follows: given the set of predicates $\Pi=\{\pi_1,...,\pi_n\}$, the algorithm iteratively calls Algorithm \ref{alg:decPred} (DecPred) in order to identify predicates whose corresponding sets have non-empty intersections. 
For each predicate $\pi_i$, we assume there exists a corresponding set $\Sc_i$ such that $\Sc_i=\{x\; | \;x\in\Re^m,\pi_i(x)=\top\}$.
The set $\Sc_i$ represents part of space $\Re^m$ where predicate function $\pi_i$ evaluates to $\top$.
When no non-empty intersection is found, the Algorithm \ref{alg:GMEP} terminates.

Algorithm \ref{alg:GMEP} creates a temporary copy of $\Pi$ in a new set $\Delta$. 
Then in a while loop Algorithm \ref{alg:decPred} is called in Line 3 to find overlapping predicates.
Algorithm DecPred checks all the combination of predicates in $\Delta$ until it finds overlapping sets (see Line 3).
The DecPred partitions two overlapping predicates $\pi_i,\pi_j$ into three mutually exclusive predicates $\overline{\pi}_{ij1},\overline{\pi}_{ij2},\overline{\pi}_{ij3}$ in Lines 4-6.
Then $\pi_i,\pi_j$ are removed from $\Delta$ in Line 7 and new predicates  $\overline{\pi}_{ij1},\overline{\pi}_{ij2},\overline{\pi}_{ij3}$ are appended to $\Delta$ (Line 8).
If no overlapping predicates are found, then DecPred returns $\emptyset$ in Line 13 and termCond gets value 1 in Line 7 of Algorithm \ref{alg:GMEP}.
Now the while loop will terminate and $\Psi$ contains all non-overlapping predicates.
We must rewrite the predicates in $\Pi$ with a disjunction operation on the new predicates in $\Psi$. 
This operation takes place in Line 11 of Algorithm \ref{alg:GMEP}. 
\yhl{Since the function CreateDisjunction is trivial}, we omit its pseudo code.
The runtime overhead of Algorithm \ref{alg:GMEP} and the size of the resulting set $\Psi$ can be exponential to $|\Pi|=n$, because we can have $2^n$ possible combinations of predicate evaluations.

\begin{figure}[t]
	\vspace{-5pt} 
	 \vspace{-10pt} 
	\begin{center}
		\includegraphics[width=12cm]{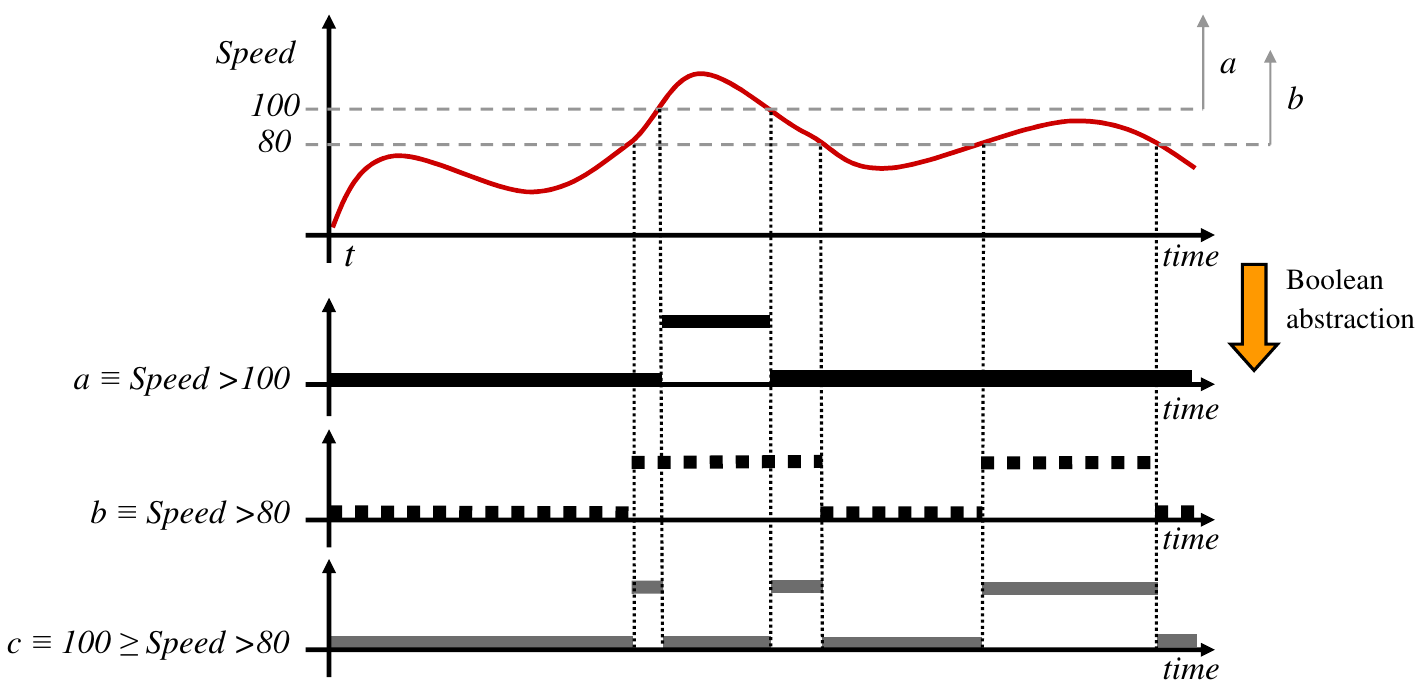}
	\end{center}
	\vspace{-10pt}
	\caption{The real-valued $Speed$ signal and its three boolean abstractions: $a\equiv speed>100$ (solid black line), $b\equiv speed>80$ (dotted line), and $c\equiv100\ge Speed>80$ (gray line).}
	\label{fig:stl2mitl}
\end{figure}

\begin{figure}
	\noindent
	\begin{minipage}{.50\linewidth}
	\vspace{-12pt}
	\begin{algorithm}[H] 
		\caption{ Generate Mutually Exclusive Predicates}
		{{\bf Input}: Set of predicates $\Pi=\{\pi_1,...,\pi_n\}$ \\
			{\bf Output}: Mutually exclusive predicates $\Psi$ \\
			Update $\Pi$ with Disjunction of predicates $\Psi$}
		\label{alg:GMEP}
		\begin{algorithmic}[1]
			\State termCond $\leftarrow$ 0; $\Delta\leftarrow\Pi$
			\While {termCond = 0}
			\State $\Psi\leftarrow$ DecPred($\Delta$)
			\If {$\Psi \neq \emptyset$}
			\State $\Delta\leftarrow \Psi$
			\Else{}
			\State termCond $\leftarrow$ 1
			\State $\Psi \leftarrow \Delta$
			\EndIf
			\EndWhile	
			\State $\Pi\leftarrow$CreateDisjunction$(\Pi,\Psi)$
			\State \Return $\Pi$,$\Psi$
		\end{algorithmic}
	\end{algorithm}
	\end{minipage}
	\vline
	\begin{minipage}{.50\linewidth}
		\vspace{-12pt}
	\begin{algorithm}[H] 
		\caption{DecPred: Decompose Two Predicates}
		{{\bf Input}: Set of predicates $\Pi=\{\pi_1,...,\pi_n\}$ \\
			{\bf Output}: Set of updated predicates $\Pi$ 
			}
		\label{alg:decPred}
		\begin{algorithmic}[1]
			\For {$i=1$ to size of $\Pi$}
			\For {$j=i + 1$ to size of $\Pi$}
			\If {$\Sc_i \cap \Sc_j$ $\neq \emptyset$ }		
			\State$\overline{\pi}_{ij1}\leftarrow\Sc_i \cap \Sc_j$
			\State$\overline{\pi}_{ij2}\leftarrow\Sc_i \setminus \Sc_j$
			\State$\overline{\pi}_{ij3}\leftarrow\Sc_j \setminus \Sc_i$ 
			\State Remove($\Pi,\{\pi_i,\pi_j\}$) 
			\State Append($\Pi,\{\overline{\pi}_{ij1},\overline{\pi}_{ij2},\overline{\pi}_{ij3}$\})
			\State \Return $\Pi$
			\EndIf
			\EndFor
			\EndFor
			\State \Return $\emptyset$
		\end{algorithmic}
	\end{algorithm}
	\end{minipage}
\end{figure}

\section{MITL Specification Debugging}
\label{vacuity}
In the following, we present algorithms that can detect inconsistency and correctness
issues in specifications. This will help the user in the elicitation of correct specifications.
Our specification debugging process conducts the following checks in this order: 1) Validity, 2) Redundancy, and 3) Vacuity. In brief, validity checking determines whether the specification is unsatisfiable or a tautology. Namely, if the specification is unsatisfiable no system can satisfy it and if it is a tautology every system can trivially satisfy it. For example, $p \vee \neg p$ is a tautology. 
If an MITL formula passes the validity checking, this means that the MITL is satisfiable but not a tautology.

Redundancy checking determines whether the specification has any redundant conjunct.
For example, in the specification $p \wedge \Box_{[0,10]} p$, the first conjunct is redundant. Sometimes redundancy is related to incomplete or erroneous requirements where the user may have wanted to specify something else. Therefore, the user should be notified.
Vacuity checking determines whether the specification has a subformula that does not affect on the satisfaction of the specification. For example, $\varphi=p\vee\Diamond_{[0,10]}p$ is vacuous since the first occurrence of $p$ does not affect on the satisfaction of $\varphi$. 
This is a logical issue because a part of the specification is overshadowed by the other components. 

The debugging process is presented in Fig. \ref{fig:debuggingFrm}. 
The feedback (Revision Necessary) to the user is a textual description about the detail of each issue.
First, given a specification, a validity check is conducted. If a formula does not pass the validity check then it means that there is a major problem in the specification and the formula is returned for revision. Therefore, redundancy and vacuity checks are not relevant at that point and the user is notified that the specification is either unsatisfiable or is a tautology.
Similarly, if the specification is redundant it means that it has a conjunct that does not have any effect on the satisfaction of the specification and we return the redundant conjunct to the user for revision. 
Lastly, if the specification is vacuous it is returned with the issue for revision by the user.
When vacuity is detected, we return to the user the simplified formula which is equivalent to the original MITL.

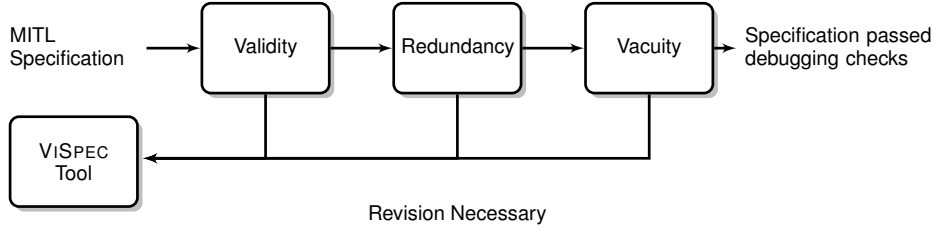
\begin{figure}
	\ifthenelse{\boolean{ARXIV}}{}{\vspace{-10pt}}
	\centering
	\scalebox{0.85}{
		\begin{tikzpicture}[font=\small\sffamily,very thick,node distance = 3cm]
		\node [frame] (ViSpec) {Validity};
		\node [left of = ViSpec,text width=2cm] (userInput) {MITL \\Specification};
		\node [frame,right of = ViSpec] (mitl) {Redundancy};
		\node [frame,right of = mitl] (AutoDeb) {Vacuity};
		\node [right of = AutoDeb,text width=3cm] (Specification) {Specification passed debugging checks};
		\node [below of = mitl, above = 0.1cm] (ret) {Revision Necessary};
		\node [frame, below of = userInput, above = 0.5cm] (return) {\textsc{ViSpec} Tool};
		\path [line] (userInput.0) |- (ViSpec);
		\path [line] (ViSpec) |- (mitl);
		\path [line] (mitl) |- (AutoDeb);
		\path [line] (AutoDeb) |- (Specification);
		\path [line] (AutoDeb) |- ($(mitl.south west) + (-3.9,-1)$); 
		\path [line] (mitl) |- ($(mitl.south west) + (-3.9,-1)$);
		\path [line] (ViSpec) |- ($(mitl.south west) + (-3.9,-1)$);
		\end{tikzpicture}
	}
	\vspace{-5pt}
	\caption{Specification Debugging}
	\label{fig:debuggingFrm}
	\vspace{-5pt}
\end{figure}

\subsection{Redundancy Checking}
Recall that a specification has a redundancy issue if one of its conjuncts can be removed without affecting the models of the specification. Before we formally present what redundant requirements are, we have to introduce some notation. We consider specification $\Phi$ as a conjunction of MITL subformulas ($\varphi_j$):
\begin{equation} \label{eq:conjunction}
\Phi=\sideset{}{_{j=1}^{k}}\bigwedge\varphi_j
\end{equation}
To simplify discussion, we will abuse notation and we will associate a conjunctive formula with the set of its conjuncts. That is:
\begin{equation} \label{eq:conjset}
\Phi=\{\varphi_j\; | \;j=1,...,k\}\equiv\varphi_1\cup\varphi_2\cup\dots\cup\varphi_k
\end{equation}
Similarly, $\{\Phi\backslash\varphi_i\}$ represents the specification $\Phi$ where the conjunct $\varphi_i$ is removed:
\ifthenelse{\boolean{ARXIV}}{
\begin{equation}\label{eq:RemoveConj}
\{\Phi\backslash\varphi_i\}=\{\varphi_j\; | \;j=1,...,i-1,i+1,...,k\}=
\sideset{}{_{j=1}^{i-1}}\bigwedge\varphi_j\wedge\sideset{}{_{j=i+1}^{k}}\bigwedge\varphi_j
\end{equation}
	}{
\begin{multline}\label{eq:RemoveConj}
\{\Phi\backslash\varphi_i\}=\{\varphi_j\; | \;j=1,...,i-1,i+1,...,k\}=
\sideset{}{_{j=1}^{i-1}}\bigwedge\varphi_j\wedge\sideset{}{_{j=i+1}^{k}}\bigwedge\varphi_j
\end{multline}
}
Therefore $\{\Phi\backslash\varphi_i\}$ represents a conjunctive formula.
Redundancy in specifications \yhl{can appear} in practice due to the incremental approach that system engineers take in the development of specifications. 
Redundancy should be avoided in formal specifications because it increases the overhead of the testing and verification processes. In addition, redundancy can be the result of incorrect translation from \yhl{natural language requirements.}
In the following, we consider the redundancy removal algorithm provided in \cite{ChocklerS09} for LTL formulas and we extend it to support MITL formulas.

\begin{definition}[Redundancy of Specification] 
A conjunct $\varphi_i$ is redundant with respect to $\Phi$ if 
\[ \underset{\psi\in\{\Phi\backslash\varphi_i\}}{\bigwedge\psi}\models \varphi_i\]
\end{definition}
 
To reformulate, $\varphi_i$ is redundant with respect to $\Phi$ if $\{\Phi\backslash\varphi_i\}\models\varphi_i$.
For example, in $\Phi=\Diamond_{[0,10]}(p\wedge q)\wedge\Diamond_{[0,10]} p\wedge\Box_{[0,10]} q$, the conjunct $\Diamond_{[0,10]}(p\wedge q)$ is redundant with respect to $\Diamond_{[0,10]} p\wedge\Box_{[0,10]} q$ since $\Diamond_{[0,10]} p\wedge\Box_{[0,10]} q\models \Diamond_{[0,10]}(p\wedge q)$. 
In addition, $\Diamond_{[0,10]} p$ is redundant with respect to $\Diamond_{[0,10]}(p\wedge q)\wedge\Box_{[0,10]} q$ since $\Diamond_{[0,10]}(p\wedge q)\wedge\Box_{[0,10]} q\models\Diamond_{[0,10]} p$. This method can \yhl{detect both issues} and report them to the user.
Algorithm \ref{alg:redun} finds redundant conjuncts in the conjunction operation of the following levels:
\begin{enumerate}
	\item Conjunction as the root formula (top level).
	\item Conjunction in the nested subformulas (lower levels).
\end{enumerate}	
In the top level, it provides the list of subformulas that are redundant with respect to the original MITL $\Phi$. 
In the lower levels, if a specification has nested conjunctive subformulas ($\phi_i\in\Phi$), it will return the conjunctive subformula $\phi_i$ as well as its redundant conjunct $\psi_j\in\phi_i$.
For example, if $\Phi=\Diamond_{[0,10]}(p \wedge \Box_{[0,10]} p)$ is checked by Algorithm \ref{alg:redun}, then it will return the pair of $(p,p \wedge \Box_{[0,10]} p)$ to represent that $p$ is redundant in $p \wedge \Box_{[0,10]} p$. 
In Line 5, the pair of $(\psi_j,\phi_i)$ is interpreted as follows: $\psi_j$ is redundant in $\phi_i$.
\begin{figure}
\noindent

\begin{minipage}{.53\linewidth}
	\vspace{-12pt}
	
	\begin{algorithm}[H]
		\fontsize{9}{9}
	\caption{Redundancy Checking}
	{\bf Input}: $\Phi$ ($MITL$ Specification)\\
	{\bf Output}: $R_\varphi$
	(redundant conjuncts w.r.t.\\ conjunctions)
	\label{alg:redun}
	\begin{algorithmic}[1]
		\State $R_\varphi\gets \emptyset$
		\For {each conjunctive subformula $\phi_i\in\Phi$}
		\For {each conjunct $\psi_j\in\phi_i$}
		\If {$\{\phi_i\backslash\psi_j\}\models\psi_j$}
		\State $R_\varphi\gets R_\varphi\cup(\psi_j,\phi_i)$
		\EndIf
		\EndFor
		\EndFor
		\State \Return $R_\varphi$
	\end{algorithmic}
	\end{algorithm}
\end{minipage}
\vline
\begin{minipage}{.47\linewidth}
	\vspace{-12pt}
	\begin{algorithm}[H]
		\fontsize{9}{9}
		\caption{ Vacuity Checking}
		{\bf Input}: $\Phi$ ($MITL$ Specification)\\
		{\bf Output}: $V_\varphi$ (vacuous formulas)\\
		\label{alg:vacui}
		\begin{algorithmic}[1]
			\State $V_\varphi\gets \emptyset$
			\For {each formula $\varphi_i\in\Phi$}
			\For {each $l\in litOccur(\varphi_i)$}
			\If{$\Phi\models\varphi_i[l\leftarrow\perp]$}
			\State $V_\varphi\gets V_\varphi\cup\{\Phi\backslash\varphi_i\}\wedge\varphi_i[l\leftarrow\perp]$
			\EndIf
			\EndFor
			\EndFor
			\State \Return $V_\varphi$
		\end{algorithmic}
	\end{algorithm}
\end{minipage}
\end{figure}

\subsection{Specification Vacuity Checking}
Vacuity detection is used to ensure that all the subformulas of the specification contribute to the satisfaction of the specification.
In other words, vacuity check enables the detection of irrelevant subformulas in the specifications \cite{ChocklerS09}.
For example, consider the STL specification $\phi_{stl} = \Diamond_{[0,10]} ( ( speed > 100 )  \vee \Diamond_{[0,10]}(speed > 80))$. 
In this case, the subformula $( speed > 100 )$ does not affect the satisfaction of the specification. 
This indicates that $\phi_{stl}$ is a vacuous specification. 
We need to create correct atomic propositions for the predicate expressions of $\phi_{stl}$ to be able to detect such vacuity issues in MITL formulas. 
If we naively replace the predicate expressions $speed > 100$ and $speed > 80$ with the atomic propositions $a$ and $b$, respectively, then the resulting MITL formula will be $\phi_{mitl} = \Diamond_{[0,10]} ( a  \vee \Diamond_{[0,10]}b)$. However, $\phi_{mitl}$ is not vacuous. Therefore, we must extract non-overlapping predicates as explained in Section \ref{stl2mitl}. The new specification $\phi_{mitl}' = \Diamond_{[0,10]} (a \vee \Diamond_{[0,10]} (a \vee c))$ where $a$ corresponds to $speed > 100$ and $c$ corresponds to $100 \geq speed > 80$ is the correct MITL formula corresponding to $\phi_{stl}$, and it is vacuous.
In the following, we provide the definition of MITL vacuity with respect to a signal:

\begin{definition}[MITL Vacuity with respect to timed trace] 
Given a timed trace $\tss$ and an MITL formula $\varphi$, a subformula $\psi$ of $\varphi$ does not affect the satisfiability of $\varphi$ with respect to $\tss$ if and only if $\psi$ can be replaced with any subformula $\theta$ without changing the satisfiability of $\varphi$ on $\tss$. 
A specification $\varphi$ is satisfied vacuously by $\tss$, denoted by $\tss\models_V\varphi$, if there exists a subformula $\psi$ which does not affect the satisfiability of $\varphi$ on $\tss$.
\end{definition}

In the following, we extend the framework presented in \cite{ChocklerS09} to support MITL specifications. Let $\varphi$ be a formula in NNF where only predicates can be in the negated form. 
A $literal$ is defined as a predicate or its negation. For a formula $\varphi$, the set of literals of $\varphi$ is denoted by $literal(\varphi)$ and contains all the literals appearing in $\varphi$. For example, if $\varphi=(\neg p\wedge q)\vee\Diamond_{[0,10]} p\vee\Box_{[0,10]} q$, then  $literal(\varphi)=\{\neg p, q, p\}$. Literal occurrences, denoted by $litOccur(\varphi)$, is a multi-set of literals appearing in some order in $\varphi$, e.g., by traversal of the parse tree. For the given example $litOccur(\varphi)=\{\neg p, q, p, q\}$. For each $l\in litOccur(\varphi)$, we create the mutation of $\varphi$ by substituting the occurrence of $l$ with $\perp$. We denote the mutated formula as $\varphi[l\leftarrow\perp]$.
\begin{definition}[$MITL$ Vacuity w.r.t. literal occurrence]
	\label{VaclitOc}
Given a timed trace $\tss$ and an $MITL$ formula $\varphi$ in NNF, specification $\varphi$ is vacuously satisfied by $\tss$ if there exists a literal occurrence $l\in litOccur(\varphi)$ such that $\tss$ satisfies the mutated formula $\varphi[l\leftarrow\perp]$. Formally, $\tss\models_V\varphi$ if  $\exists l \in litOccur(\varphi)$ s.t. $\tss\models\varphi[l\leftarrow\perp]$.
\end{definition}
\begin{theorem}[$MITL$ Inherent Vacuity] 
	\label{thm:vacu}
Assume that the specification $\Phi$ is a conjunction of MITL formulas. If  $\exists\varphi_i\in\Phi$ and $\exists l\in litOccur(\varphi_i)$, such that $\Phi\models\varphi_i[l\leftarrow\perp]$, then $\Phi$ is inherently vacuous.
\end{theorem}

\bhl{A specification $\Phi$ is {\it inherently vacuous} if it is equivalent to its simplified mutation, which means that $\Phi$ is vacuous independent of any signal or system.
Inherent vacuity of LTL formulas is addressed in \cite{FismanKSV08,ChocklerS09}.}
The proof of \textsc{Theorem} \ref{thm:vacu} is straightforward modification of the proofs given in \cite{ChocklerS09,KupfermanV03}. For completeness in the presentation, we provide the proof in Appendix \ref{app:Thm1}.
When we do not have a root-level conjunction in the specification ($\Phi=\varphi$\footnote{In this case, we assume $\{\Phi\backslash\varphi\}\equiv\top$ in Line 5 of the Algorithm 4.}), we check the vacuity of the formula with respect to itself. In other words, we check whether the specification satisfies its mutation ($\varphi\models\varphi[l\leftarrow\perp]$). 
Technically, Algorithm 4 as presented, returns a list of all the mutated formulas that are equivalent to the original MITL.

\section{Signal Vacuity Checking}
\label{svc}
In the previous section, we addressed specification vacuity without considering the system.
However, in many cases, specification vacuity depends on the system. 
For example, consider the LTL specification $\varphi=\Box(req\Rightarrow\Diamond ack)$. 
The specification $\varphi$ does not have an inherent vacuity issue \cite{FismanKSV08}.
However, if $req$ never happens in any of the behaviors of the system, then the specification $\varphi$ is vacuously satisfied on this specific system.
As a result, it has been argued that it is important to add vacuity detection in the model checking process \cite{BeerBER01,KupfermanV03}.
We encounter the same issue when we test signals and systems with respect to Request-Response STL/MITL specifications. 

\subsection{Vacuous Signals}
\label{sec:AF}
Consider the MITL specification $\varphi=\Box_{[0,5]}(req\Rightarrow\Diamond_{[0,10]}ack)$. This formula will pass the MITL Specification Debugging method presented in Section \ref{vacuity}. However, any timed trace $\mu$ that does not satisfy $req$ at any point in time during the test will vacuously satisfy $\varphi$. 
We refer to timed traces that do not satisfy the antecedent (precondition) of the subformula as {\it vacuous timed traces}. 
Similarly, these issues follow for STL formulas over signals as well.
Consider Task 6 in Table \ref{tab:taskList} with the specification $\psi=\Diamond_{[0,40]}(speed>100)\Rightarrow\Box_{[0,30]}(rpm>3000)$. Any real-valued signal $s$ that does not satisfy $\Diamond_{[0,40]}(speed>100)$ will vacuously satisfy $\psi$. Finding such signals is important in testing and monitoring, since if a signal $s$ does not satisfy the precondition of an STL/MITL specification, then there is no point in considering $s$ as a useful test. 
\begin{definition}[Vacuous Timed Trace (Signal)] 
	\label{def:vacSig}
\bhl{Given an MITL (STL) formula $\varphi$, a timed trace $\tss$ (signal $s$) is vacuous if it satisfies the Antecedent Failure mutation of $\varphi$.}
\end{definition}
\bhl{Antecedent Failure is one of the main sources of vacuity. Antecedent Failure occurs in a Request-Response specification such as $\varphi_{RR}=\Box_{[0,5]}(req\Rightarrow\Diamond_{[0,10]}ack)$.} 
We provide a formula mutation that can detect signal vacuity in Request-Response specifications.
\begin{definition}[Request-Response MITL]
A Request-Response MITL formula $\varphi_{RR}$ is an MITL formula that has one or more implication ($\Rightarrow$) operations in positive polarity (without any negation). In addition, for each implication operation the \bhl{consequent must have a temporal operator at the top-level.}
\end{definition}

In the Request-Response (RR) specifications \cite{HornTW015}, we define sequential events in a specific order (by using the implication operator). 
Many practical specification patterns based on the Request-Response format are provided for system properties \cite{Dwyer1998PSP,Konrad2005RSP}. 
Therefore, we can define a chain of events that the system must respond/react to.
In an RR-specification such as $\varphi_{RR}=\Box_{[0,5]}(req\Rightarrow\Diamond_{[0,10]}ack)$, the temporal
operator for the consequent $\Diamond_{[0,10]}ack$ is necessary, unless the system does not have any time to acknowledge the $req$.
For any trace $\mu$ in which $req$ never happens, we can substitute $ack$ by any formula and the specification is still satisfied by $\mu$. Therefore, if the antecedent is failed by a trace $\mu$, then $\varphi_{RR}$ is vacuously satisfied by $\mu$.
For each implication subformula ($\varphi\Rightarrow\psi$), the left operand ($\varphi$) is the precondition (antecedent) of the implication. An antecedent failure mutation is a new formula that is created with the assertion that the precondition ($\varphi$) never happens. 
Note that RR-specifications should not be translated into NNF.
For each precondition $\varphi$, we create an antecedent failure mutation $\Box_{I_{\varphi}}(\neg\varphi)$ where $I_{\varphi}$ is called the {\it effective interval} of $\varphi$.
	
\begin{definition}[Effective Interval] 	
The effective interval of a subformula is the time interval when the subformula can have an impact on the truth value of the whole MITL (STL) specification.
\end{definition}	
Each subformula is evaluated only in the time window that is provided by the {\it effective interval}.
For example, for \bhl{the} MITL specification $\varphi\wedge\psi$, the effective interval for both $\varphi$ and $\psi$ is [0,0], because $\varphi$ and $\psi$ can change the value of $\varphi\wedge\psi$ only within the interval [0,0]. 
Similarly, for the MITL specification $\Box_{[0,10]}\varphi$, the effective interval of $\varphi$ is [0,10], since the truth value of $\varphi$ is observed in the time window of [0,10] for evaluating $\Box_{[0,10]}\varphi$.
The effective interval is important for the creation of an accurate antecedent failure mutation. This is because the antecedent can affect the truth value of the MITL formula only if it is evaluated in the effective interval. 
The effective interval is like a time window to make the antecedent observable for an outside observer the way it is observed by the MITL specification.
\begin{algorithm}[t]
	\caption{Effective Interval Update EIU($\varphi$,$I$)}
	{\bf Input}:  $\varphi$ (Parse Tree of the MITL formula), $I$ (Effective Interval)\\
	{\bf Output}: $\varphi$ (Updated formula with subformulas annotated with effective intervals)
	\label{alg:AIcom}
	
	\begin{algorithmic}[1]
		\State	$\varphi.EI\gets I$
		\If {$\varphi\equiv\neg\varphi_m$}
		\State EIU($\varphi_m$,$I$)
		\ElsIf {$\varphi\equiv\varphi_m\vee\varphi_n$ OR $\varphi\equiv\varphi_m\wedge\varphi_n$ OR $\varphi\equiv\varphi_m\Rightarrow\varphi_n$ }
		\State EIU($\varphi_m$,$I$)
		\State EIU($\varphi_n$,$I$)
		\ElsIf{$\varphi\equiv\Box_{I'}\varphi_m$ OR $\varphi\equiv\Diamond_{I'}\varphi_m$ }
		\State	  $I''\gets I' \oplus I  $
		\State EIU($\varphi_m$,$I''$)
		\EndIf
		\State \Return $\varphi$
	\end{algorithmic}
\end{algorithm}

The effective interval of MITL formulas can be computed recursively using Algorithm \ref{alg:AIcom}. To run Algorithm \ref{alg:AIcom}, we must process the MITL formula parse tree\footnote{We assume that the MITL specification is saved in a binary tree data structure where each node is a formula with the left/right child as the left/right corresponding subformula of $\varphi$. In addition, we assume that the nodes of  $\varphi$'s tree contain a field called $EI$ where we annotate the effective interval of $\varphi$ in $EI$, namely, $\varphi.EI\gets I_{\varphi}$.}.
The algorithm must be initialized with the interval of [0,0] for the top node of the MITL formula, namely, EIU($\varphi$,[0,0]). 
This is because, according to the semantics of MITL, the value of the whole MITL formula is only important at time zero. In Line 8 of Algorithm \ref{alg:AIcom}, the operator $\oplus$ is used to add two intervals as follows: 
\begin{definition}[$\oplus$] 
\label{add_interval}
Given intervals $I=[l,u]$ and $I'=[l',u']$, we define $I''\gets I  \oplus I'$ where $I''=[l'',u'']$ such that $l''=l+l'$ and $u''=u+u'$. 
\end{definition}
If either $I$ or $I'$ is left open (resp. right open), then $I''$ will be left open (resp. right open)\footnote{Although we assume in Definition \ref{def:syn} that intervals are right-closed, Algorithm \ref{alg:AIcom} can be applied to right-open intervals as well.}. 
In Line 1 of Algorithm \ref{alg:AIcom}, the input interval $I$ is assigned to the effective interval of $\varphi$, namely $\varphi.EI$. If the top operation of $\varphi$ is a propositional operation ($\neg,\vee,\wedge,\Rightarrow$) then the $I$ will be propagated to subformulas of $\varphi$ (see Lines 2-6). If the top operation of $\varphi$ is a temporal operator ($\Box_{I'},\Diamond_{I'}$), then the effective interval is modified according to Definition \ref{add_interval} and the interval $I''\gets I  \oplus I'$ is propagated to the subformulas of $\varphi$. 
For example, assume that the MITL specification is $\varphi_{RR}=\Box_{[1,2]}(\Diamond_{[3,5]}b\Rightarrow(\Box_{[4,6]}(c\Rightarrow \Diamond_{[0,2]}d)))$. The specification $\varphi$ has two antecedents, $\alpha_1=\Diamond_{[3,5]}b$ and $\alpha_2=c$. The effective intervals of $\alpha_1$ and $\alpha_2$ are $I_{\alpha_1}=[0,0]\oplus[1,2]=[1,2]$ and $I_{\alpha_2}=[0,0]\oplus[1,2]\oplus[4,6]=[5,8]$, respectively. As a result, the antecedent failure mutations are $\Box_{[1,2]}(\neg\Diamond_{[3,5]}b)$ and $\Box_{[5,8]}(\neg c)$, respectively.
Algorithm \ref{alg:antec} returns the list of antecedent failures $AF_\varphi$, namely all the vacuously satisfied implication subformulas by $s$. If the $AF_\varphi$ list is empty, then the signal $s$ is not vacuous.
\bhl{To check whether the signal $s$ satisfies $\varphi$'s mutations in Algorithm \ref{alg:antec} (Line 5), we should use an off-line monitor such as \cite{Fainekos2012}.}
\begin{figure}
	\noindent
	\begin{minipage}{.51\linewidth}
		\vspace{-12pt}
		\begin{algorithm}[H]
			\caption{Antecedent Failure}
			{\bf Input}: $\varphi_{RR}$,$s$ (RR-Specification, Signal)\\
			{\bf Output}: $AF_\varphi$ a list of failed antecedents
			\label{alg:antec}
			\begin{algorithmic}[1]
				\State $AF_\varphi\gets \emptyset$
				\State EIU($\varphi$,[0,0])
				\For {each implication $(\varphi_i\Rightarrow\psi_i)\in\varphi_{RR}$}
				\State  $I\varphi_i\gets\varphi_i.EI$
				\If{$s\models\Box_{I\varphi_i}(\neg\varphi_i)$}
				\State $AF_\varphi\gets AF_\varphi\cup(\varphi_i\Rightarrow\psi_i)$
				\EndIf
				\EndFor
				\State \Return $AF_\varphi$
			\end{algorithmic}
		\end{algorithm}
	\end{minipage}
	\vline
	\begin{minipage}{.49\linewidth}
		\vspace{-12pt}
		\begin{algorithm}[H]
			\caption{Literal Occurrence Removal}
			{\bf Input}: $\Phi$,$s$ (Specification, Signal)\\
			{\bf Output}: $MF_\varphi$ a list of mutated formulas
			\label{alg:litOcRe}
			\begin{algorithmic}[1]
				\State $MF_\varphi\gets \emptyset$
				\For {each formula $\varphi_i\in\Phi$}
				\For {each $l\in litOccur(\varphi_i)$}
				\If{$s\models\varphi_i[l\leftarrow\perp]$}
				\State $MF_\varphi\gets MF_\varphi\cup\{\varphi_i[l\leftarrow\perp]\}$
				\EndIf
				\EndFor
				\EndFor
				\State \Return $MF_\varphi$
			\end{algorithmic}
		\end{algorithm}
	\end{minipage}
\end{figure}

\subsection{Vacuity Detection in Testing and Falsification}
\label{VAT}
Detecting vacuous satisfaction of specifications is usually applied on top of model checking tools for finite state systems \cite{BeerBER01,KupfermanV03}. 
However, in general, the verification problem for hybrid automata (a mathematical model of CPS) is undecidable \cite{hs_gf:Alur95algorithmic}. 
Therefore, a formal guarantee about the correctness of CPS modeling and design is impossible, in general.
 CPS are usually safety critical systems and the verification and validation of these systems is necessary. 
One approach is to use Model Based Design (MBD) with a mathematical model of the CPS to facilitate the system analysis and implementation \cite{AbbasTECS2013}. 
Thus, semi-formal verification methods are gaining popularity \cite{KapinskiDJIB15}. 
Although we cannot solve the correctness problem with testing and monitoring, we can detect possible errors with respect to STL requirements. 

In Fig. \ref{fig:flow}, a testing approach for signal vacuity detection is presented. 
The input generator creates initial conditions and inputs to the system under test. The system under test can be a Model, Processor in the Loop (PiL), Hardware in the Loop (HiL) or a real system. 
An example of a test generation technology that implements the architecture in Fig. \ref{fig:flow} is presented in \cite{AbbasTECS2013}.
The system under test is simulated to generate an output trace. Then, a monitor checks the trace with respect to the specification and reports to the user whether the system trace satisfies or falsifies the specification (for example \cite{Maler2004}). 
For each falsification, we will report to the user the falsifying trajectory to investigate the system for this error.
Falsification based approaches for CPS can help us find subtle bugs in industrial size control systems \cite{JinEtAl13hscc}. 
If after using stochastic-based testing and numerical analysis we could not find any bugs, then we are more confident that the system works correctly.
However, it will be \bhl{concerning if} the numerical analysis is mostly based on vacuous signals. \bhl{This is because} vacuous signals satisfy the specification for reasons other than what was \bhl{originally} intended.
Signal vacuity checking is conducted in Fig. \ref{fig:flow} using Algorithm \ref{alg:antec}, and vacuous signals are reported to the user for further inspection. 
This will help users to focus \bhl{their analysis} on the part of the system that \yhl{generates vacuous signals to prevent vacuous test generation.}
\subsubsection {Detecting Partially Covering Signals} \label{sec:LOR}
\bhl{
A problem closely related to vacuity detection is the partial coverage problem. In this section, we show that Literal Occurrence Removal can be used for determining partially covering signals. 
Partially covering signals are the signals that not only satisfy the specification but also they satisfy Literal Occurrence Removal mutation:} 
\begin{definition}[Partially Covering Timed Trace (Signal)] 
	\label{def:covSig}
	\bhl{Given an MITL (STL) formula $\varphi$ in NNF, a timed trace $\tss$ (signal $s$) is partially covering if it satisfies the Literal Occurrence Removal mutation of $\varphi$.}
\end{definition}
This mutation is generated by repeatedly substituting the occurrences of literals with $\perp$\bhl{, which is} denoted by $\varphi[l\leftarrow\perp]$ (see the Definition \ref{VaclitOc}). 
In Algorithm \ref{alg:litOcRe}, we check whether the signal will satisfy the mutated specification ($\varphi_i[l\leftarrow\perp]$). 
\bhl{In the following, we prove that all satisfying signals are also partially covering signals:}

\begin{theorem}
	\label{thm:sigVacuity} 
	For all MITL formula $\varphi\in\Phi$ in NNF, if there exists a disjunction subformula in $\varphi$, then for all $\tss$ such that $\tss\models\varphi$, it is always the case that \bhl{there exists a literal $l\in litOccur(\varphi)$ s.t. $\tss\models\varphi[l\leftarrow\perp]$.}
\end{theorem}
\bhl{The proof of Theorem \ref{thm:sigVacuity} is provided in Appendix \ref{app:SV}. For any MITL (STL) specification $\varphi$, which contains one or more disjunction operators ($\vee$) in NNF, any timed trace (signal) that satisfies $\varphi$ will also satisfy a mutation $\varphi[l\leftarrow\perp]$ for some literal occurrence $l$.}
Further, any specification which lacks a disjunction operator ($\vee$) in its NNF will not satisfy $\varphi[l\leftarrow\perp]$ for any literal occurrence $l$. 
That is, for formulas without any disjunction operator in NNF, we have
$\varphi[l\leftarrow\perp]\equiv\perp$ since for any MITL formula $\varphi$, we have $\varphi\wedge\bot\equiv\bot$. 
\bhl{If all satisfying signals are partially covering signals, what is the benefit of Algorithm \ref{alg:litOcRe}?
	There are two applications for Algorithm \ref{alg:litOcRe} as follows.}

First, Algorithm \ref{alg:litOcRe} can find which (how many) disjuncts are satisfied by the partially covering signal.
This information can be \bhl{used by the falsification technique} to find the disjuncts/predicates that cause the \bhl{formula} satisfaction.
\yhl{Therefore, the falsification method will target the system behaviors corresponding to those predicates.}
As a result, Algorithm \ref{alg:litOcRe} can be used to \bhl{improve the falsification method}.

Second, Algorithm \ref{alg:litOcRe} can \bhl{also} be used for coverage analysis when falsification occurs.
This is because with a slight modification, the dual of Algorithm \ref{alg:litOcRe} can help us to find the source of the falsification.
According to Corollary \ref{cor:con}, for any $\varphi$ in NNF, where $\varphi$ has a conjunctive subformula of $\psi=\psi_1\wedge\psi_2$, if $\tss\not\models\varphi$ then $\exists l\in litOccur(\psi)$ s.t. $\tss\not\models\varphi[l\leftarrow\top]$.
Now, we can identify which conjunct of $\psi$ contributes towards the falsification by substituting it by iteratively applying $\varphi[l\leftarrow\top]$. This can be better explained in the following example:
\begin{example}
	Assume $\varphi=\Diamond_{I_1}(a\wedge\Diamond_{I_2}b)$ and a falsifying trace $\tss\not\models\varphi$ exists. Formula $\varphi$ contains conjunction of $\psi=a\wedge\Diamond_{I_2}b$ and  $litOccur(\psi)=\{a,b\}$. We can substitute $a$ and $b$ with $\top$ to find the main source of falsification of $\varphi$ as follows: \\
	$\bullet$ If $\tss\models\varphi[a\leftarrow\top]$ then $\tss\models\Diamond_{I_1}(\Diamond_{I_2}b)$, so $a$ is the source of the problem.\\
	$\bullet$ If $\tss\models\varphi[b\leftarrow\top]$ then $\tss\models\Diamond_{I_1}(a)$, so $b$ is the source of the problem.
\end{example}
As a result, the dual of Algorithm \ref{alg:litOcRe} can be used to debug a trace when the counter example is created using falsification methodologies \cite{AbbasTECS2013}.

\begin{figure}
	\noindent
	\begin{minipage}{.60\linewidth}
		\vspace{-10pt} 
		\centering
		\includegraphics[width=8cm]{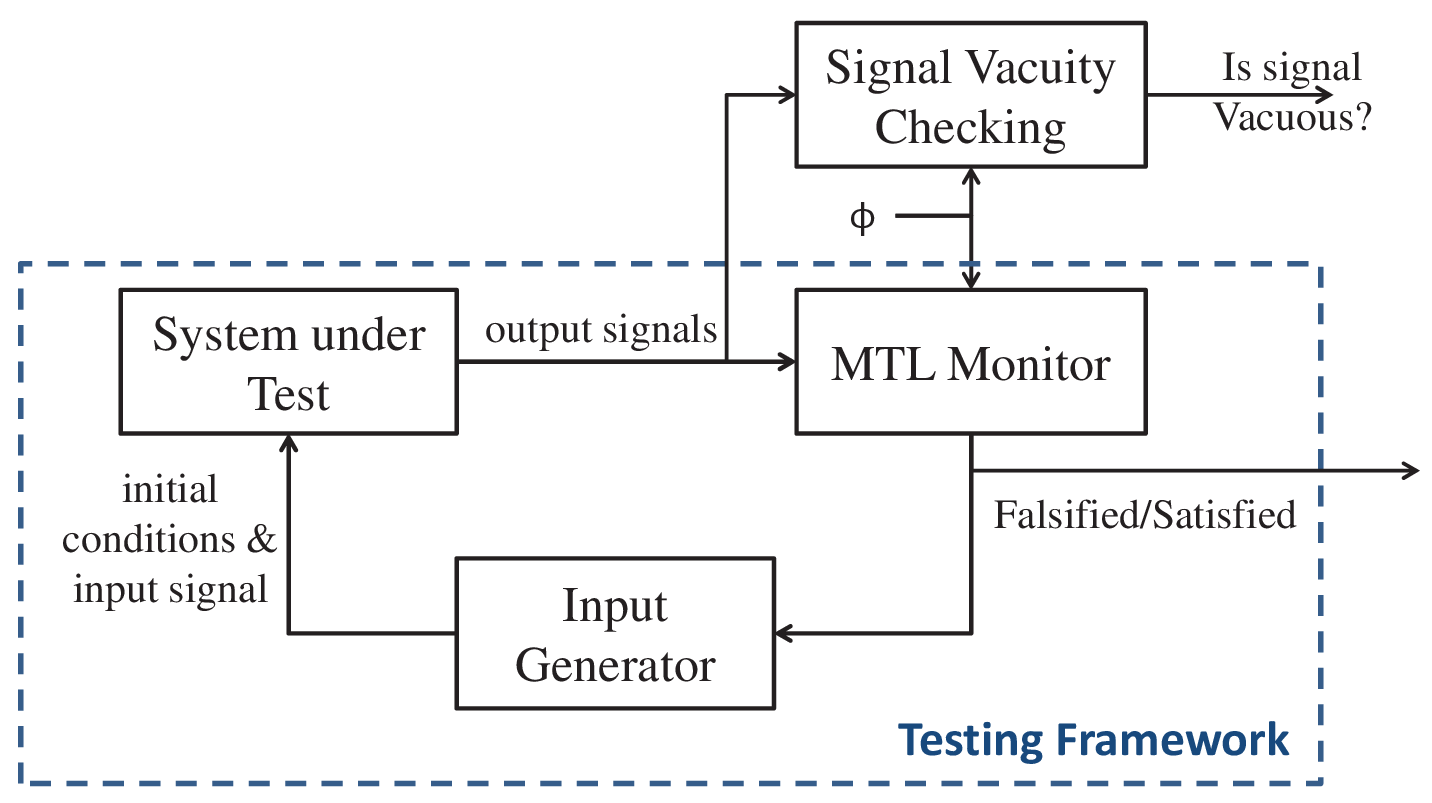}
		\caption{Using signal vacuity checking to improve the confidence of an automatic test generation framework.}
		\label{fig:flow}
		\vspace{-5pt} 
	\end{minipage}
	\begin{minipage}{.40\linewidth}
		\vspace{-10pt} 
		\centering
		\includegraphics[width=5.5cm]{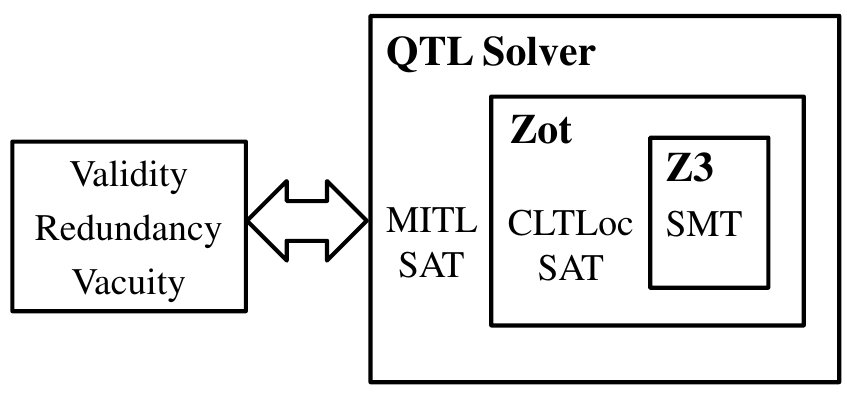}
		\caption{The MITL SAT solver from \cite{BersaniRP16} is used for debugging specifications.}
		\label{fig:MITLSAT}
		\vspace{-5pt} 
	\end{minipage}
\end{figure}

\section{Experimental Analysis}
\label{experiments}
\yhl{All three levels of the correctness analysis} of MITL specifications need satisfiability checking as the underlying tool \cite{MITLSat13}.
In validity checking, we simply check whether the specification and its negation are satisfiable.
In general, in order to check whether $\varphi\models\psi$, we should check whether $\varphi\Rightarrow\psi$ is a tautology, that is $\forall\tss,\tss\models\varphi\Rightarrow\psi$.
This can be verified by checking whether $\neg(\varphi\Rightarrow\psi)$ is unsatisfiable.
Recall that $\varphi\Rightarrow\psi$ is equivalent to $\neg\varphi\vee\psi$.
So we have to check whether $\varphi\wedge\neg\psi$ is unsatisfiable to conclude that $\varphi\models\psi$.
We use the above reasoning for redundancy checking as well as for vacuity checking. 
For \yhl{redundancy checking of conjuncts at the root level}, $\{\Phi\backslash\varphi_i\}\wedge\neg\varphi_i$ should be unsatisfiable, in order to conclude that $\{\Phi\backslash\varphi_i\}\models\varphi_i$. 
For vacuity checking, $\Phi\wedge\neg(\varphi_i[l\leftarrow\perp])$ should be unsatisfiable, in order to prove that $\Phi\models\varphi_i[l\leftarrow\perp]$.

\subsection{MITL Satisfiability}
\label{MITLSAT}
The satisfiability problem of MITL is \yhl{EXPSPACE-complete} \cite{AlurFH96}.
In order to check whether an MITL formula is satisfiable we use two publicly available tools: qtlsolver\footnote{qtlsolver:
A solver for checking satisfiability of Quantitative / Metric Interval Temporal Logic (MITL/QTL) over Reals. Available from \url{https://code.google.com/p/qtlsolver/}} and zot\footnote{The zot bounded model/satisfiability checker. Available from \url{https://code.google.com/p/zot/}}. The qtlsolver that we used translates MITL formulas into CLTL-over-clocks \cite{BersaniRP16,MITLSat13}. Constraint LTL (CLTL) is an extension of LTL where predicates \yhl{are allowed to be assertions} on the values of non-Boolean variables \cite{DemriD07}. That is, in CLTL, we are allowed to define predicates using relational operators for variables over domains like $\Ne$ and $\Ze$. Although satisfiability of CLTL in general is not decidable, some variants of it are decidable \cite{DemriD07}. 

CLTLoc (CLTL-over-clocks) is a variant of CLTL where the clock variables are the only arithmetic variables that are considered in the atomic constraints. It has been proven in \cite{BersaniRP14} that CLTLoc is equivalent to timed automata \cite{mc_gf:CGP99}. Moreover, it can be polynomially reduced to decidable Satisfiability Modulo Theories which are solvable by many SMT solvers such as Z3\footnote{Microsoft Research, Z3: An efficient SMT solver. Available from \url{http://research.microsoft.com/en-us/um/redmond/projects/z3/}}. The satisfiability of CLTLoc is PSPACE-complete \cite{MITLSat13} and the translation from MITL to CLTLoc in the worst case can be exponential \cite{BersaniRP16}. 
Some restrictions must be imposed on the MITL formulas in order to use the qtlsolver \cite{BersaniRP16}. That is, the lower bound and upper bound for the intervals of MITL formulas should be integer values and the intervals are left/right closed. 
Therefore, we expect the values to be integer when we analyse MITL formulas. 
The high level architecture of the MITL SAT solver, which we use to check the three issues, is provided in Fig. \ref{fig:MITLSAT}.
\subsection{Specification Debugging Results}
We utilize the debugging algorithm on a set of specifications developed as part of a usability study for the evaluation of the \textsc{ViSpec} tool \cite{Hoxha_ViSpecIROS15}. The usability study was conducted on two groups:
\vspace{-5pt} 
\begin{enumerate}
\item Group A: These are users who declared that they have little to no experience in working with requirements. \yhl{The Group A} cohort consists of twenty subjects from the academic community at Arizona State University. 
Most of the subjects have an engineering background. 
\item Group B: These are users who declared that they have experience working with system requirements. Note that they do not necessarily have experience in writing requirements using formal logics. \yhl{The Group B} subject cohort was comprised of ten \yhl{subjects from industry} in the Phoenix metro area. 
\end{enumerate}
\vspace{-5pt} 
Each subject received a task list to complete. The list contained ten tasks related to automotive system specifications. Each task asked the subject to formalize a natural language specification through \textsc{ViSpec} and generate an STL specification. The task list is presented in Table \ref{tab:taskList}. A detailed report on the accuracy of the users response to each natural language requirement is provided in \cite{Hoxha_ViSpecIROS15}. Note that the specifications were preprocessed and transformed from the original STL formulas to MITL in order to run the debugging algorithm. For example, specification $\phi_3$ in Table \ref{tab:vacResults} originally in STL was $\phi_{3_{STL}}=\Diamond_{[0,40]}((( speed > 80 ) \Rightarrow  \Diamond_{[0,20]}(rpm > 4000) )\wedge \Box_{[0,30]} (speed > 100) )$. The STL predicate expressions $( speed > 80 ),(rpm > 4000),(speed > 100) $ are mapped into atomic propositions with non-overlapping predicates (Boolean functions) $p_1,p_2,p_3$. The predicates $p_1,p_2,p_3$ correspond to the following STL representations: $p_1 \equiv speed>100$, $p_2 \equiv rpm>4000$, and $p_3 \equiv 100\ge speed>80$.
In Table \ref{tab:vacResults}, we present the common issues with the elicited specifications that our debugging algorithm detects. 
Note that validity, redundancy and vacuity issues are present in the specifications listed. 
It should be noted that for specification $\phi_3$, although finding the error takes a significant amount of time, our algorithm can be used off-line. 

In Fig. \ref{fig:debuggingRuntimeOverhead}, we present the runtime overhead of the three stage debugging algorithm over specifications collected in the usability study. 
In the first stage, 87 specifications go through validity checking. 
Five specifications fail the test and therefore they are immediately returned to the user.  
As a result, 82 specifications go through redundancy checking of conjunction in the root level \footnote{In these experiments, we did not consider conjunctions in the lower level subformulas for redundancy checking.}, where 9 fail the test.
Lastly, 73 specifications go through vacuity checking \yhl{where 5 specifications have vacuity issues.} The remaining 68 specifications passed the tests. Note that in the figure, two outlier data points are omitted from the vacuity sub-figure for presentation purposes. The two cases were timed at 39,618sec and 17,421sec. In both cases, the runtime overhead was mainly because the zot software took hours to determine that the modified specification is unsatisfiable (both specifications were vacuous). The overall runtime of $\phi_3$ in Table \ref{tab:vacResults} is 39,645sec which includes the runtime of validity and redundancy checking. The runtime overhead of vacuity checking of $\phi_3$ can be reduced by half because, originally, in vacuity checking we run MITL satisfiability checking for all literal occurrences. In particular, $\phi_3$ has four literal occurrences where for two cases zot took more than 19,500sec to determine that the modified specification is unsatisfiable. We can provide an option for early detection: stop and report as soon as an issue is found (the first unsatisfiability).

The circles in Fig. \ref{fig:debuggingRuntimeOverhead} represent the timing performance in each test categorized by the number of literal occurrences and temporal operators. The asterisks represent the mean values and the dashed line is the linear interpolation between them.
In general, we observe an increase in the average computation time as the number of literal occurrences and temporal operators increases. 
All the experimental results in Section \ref{experiments} were performed on an Intel Xeon X5647 (2.993GHz) with 12 GB RAM.
\begin{table*}[tb]
  \centering
  \ifthenelse{\boolean{ARXIV}}{\caption{Task list with automotive system specifications presented in natural language\label{tab:taskList}}}{\tbl{Task list with automotive system specifications presented in natural language\label{tab:taskList}}}
  {
    \begin{tabular}{p{3.2cm}p{10cm}}
    \toprule
    Task  & Natural Language Specification \\
    \midrule
    1.   Safety  & In the first 40 seconds, vehicle speed should always be less than 160. \\
    2.   Reachability & In the first 30 seconds, vehicle speed should go over 120. \\
    3.   Stabilization & At some point in time in the first 30 seconds, vehicle speed will go over 100 and stay above for 20 seconds. \\
    4.   Oscillation & At every point in time in the first 40 seconds, vehicle speed will go over 100 in the next 10 seconds. \\
    5.   Oscillation & It is not the case that, for up to 40 seconds, the vehicle speed will go over 100 in every 10 second period.  \\
    6.   Implication & If, within 40 seconds, vehicle speed is above 100 then within 30 seconds from time 0, engine speed should be over 3000. \\
    7.   Request-Response & If, at some point in time in the first 40 seconds, vehicle speed goes over 80 then from that point on, for the next 30 seconds, engine speed should be over 4000. \\
    8.   Conjunction & In the first 40 seconds, vehicle speed should be less than 100 and engine speed should be under 4000. \\
    9.   Non-strict sequencing & At some point in time in the first 40 seconds, vehicle speed should go over 80 and then from that point on, for the next 30 seconds, engine speed should be over 4000. \\
    10. Long sequence & If, at some point in time in the first 40 seconds, vehicle speed goes over 80 then from that point on, if within the next 20 seconds the engine speed goes over 4000, then, for the next 30 seconds, the vehicle speed should be over 100. \\
    \bottomrule
    \end{tabular}}
\end{table*}

\begin{table*}[t]
\centering
\ifthenelse{\boolean{ARXIV}}{\caption{Incorrect specifications from the usability study in \cite{Hoxha_ViSpecIROS15}, error reported to the user by the debugging algorithm, and algorithm runtime. Formulas have been translated from STL to MITL.\label{tab:vacResults}}}{\tbl{Incorrect specifications from the usability study in \cite{Hoxha_ViSpecIROS15}, error reported to the user by the debugging algorithm, and algorithm runtime. Formulas have been translated from STL to MITL.\label{tab:vacResults}}} {
  \begin{tabular}{| c | c | c | c |c |}
    \hline
   $\phi$ & Task \# & MITL Specification created by \textsc{ViSpec} users & Reporting the errors & Sec.\\ \hline\hline
   $\phi_1$ & 3& $\Diamond_{[0,30]} p_1 \wedge \Diamond_{[0,20]} p_1$ & $\Diamond_{[0,30]} p_1$ is redundant & 14 \\ \hline 
 $\phi_2$ & 3 & $\Diamond_{[0,30]}(p_1\Rightarrow\Box_{[0,20]} p_1 ) $& $\varphi$ is a tautology & 7 \\ \hline  
   $\phi_3$ & 10& $\Diamond_{[0,40]}((( p_1 \vee p_3 ) \Rightarrow \Diamond_{[0,20]} p_2 )\wedge \Box_{[0,30]} p_1 )$ & $\varphi$ is vacuous: $\varphi\models\varphi[p_3\leftarrow\perp]$ & 39645 \\ \hline 
  $\phi_4$ & 4 & $\Box_{[0,40]} p_1 \wedge \Box_{[0,40]} \Diamond_{[0,10]} p_1$ & $\Box_{[0,40]} \Diamond_{[0,10]} p_1$ is redundant& 29 \\ 
        \hline
$\phi_5$ &  10 & $\Diamond_{[0,40]}( p_1 \vee p_3) \wedge \Diamond_{[0,40]} p_2 \wedge \Diamond_{[0,40]} \Box_{[0,30]} p_1 $
 & $\Diamond_{[0,40]}( p_1 \vee p_3)$ is redundant& 126 \\ 
         \hline
    \end{tabular}}
\end{table*}

\begin{figure*}
\centering
\begin{tabular}{c}
Validity \\
\includegraphics[width=13cm]{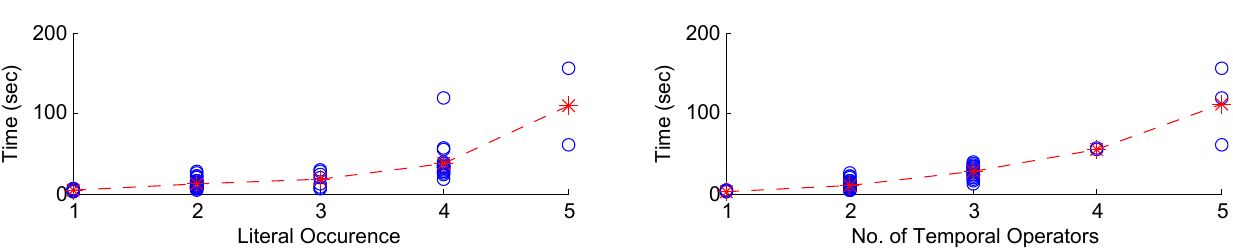} \\
Redundancy \\
\includegraphics[width=13cm]{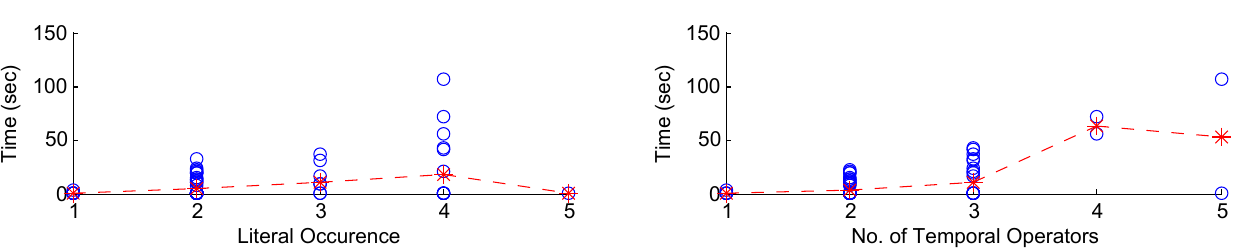} \\
Vacuity \\
\includegraphics[width=13cm]{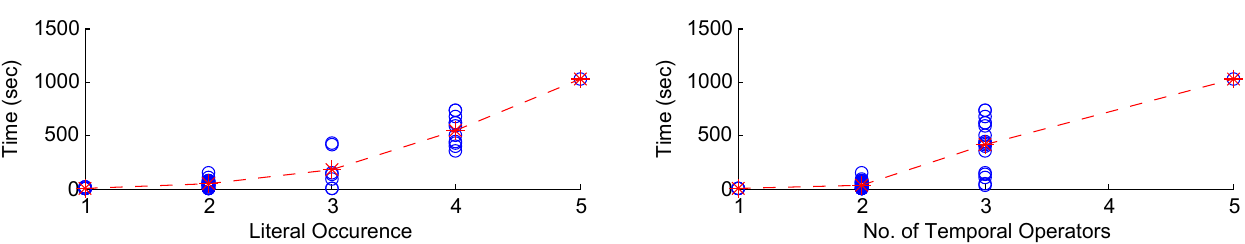} 
\end{tabular}
\caption{Runtime overhead of the three stages of the debugging algorithm over user-submitted specifications. Timing results are presented over the number of literal occurrences and the number of temporal operators.}
\label{fig:debuggingRuntimeOverhead}
\end{figure*}
\subsection{LTL Satisfiability}
\label{ltlsat}
In the previous section, we mentioned \yhl{that MITL satisfiability} is a computationally hard problem. However, in practice, we know that LTL satisfiability is solvable faster \yhl{than MITL satisfiability} \cite{LiZPVH13}. In this section, we consider how we can use the satisfiability of LTL formulas to decide about the satisfiability of MITL formulas. Consider the following fragments of MITL and LTL in NNF:\\
MITL($\Box$):
$\varphi\;::=\;\top \;|\; \bot \; | \; p \; | \; \neg p \; | \; \varphi _1 \wedge \varphi_2\; | \; \varphi _1 \vee \varphi_2  \; | \; \Box_I\varphi _1 $\\
MITL($\Diamond$):  $\varphi\;::=\;\top \;|\; \bot \; | \; p \; | \; \neg p \; | \; \varphi _1 \wedge \varphi_2\; | \; \varphi _1 \vee \varphi_2 \; | \; \Diamond_I\varphi_1$\\
LTL($\Box$):  $\varphi\;::=\;\top \;|\; \bot \; | \; p \; | \; \neg p \; | \; \varphi _1 \wedge \varphi_2\; | \; \varphi _1 \vee \varphi_2  \; | \; \Box\varphi _1 $\\
LTL($\Diamond$):  $\varphi\;::=\;\top \;|\; \bot \; | \; p \; | \; \neg p \; | \; \varphi _1 \wedge \varphi_2\; | \; \varphi _1 \vee \varphi_2 \; | \; \Diamond\varphi_1$

In Appendix \ref{app:MITLvsLTL}, we prove that the satisfaction of a formula $\phi_M\in$ MITL($\Diamond$) in NNF is related to the satisfaction of an LTL version of $\phi_M$ called $\phi_L\in$ LTL($\Diamond$) where $\phi_L$ is identical to $\phi_M$ except \yhl{that every interval} $I$ in $\phi_M$ is removed. For example, if $\phi_M=\Diamond_{[0,10]}(p\wedge q)\wedge\Diamond_{[0,10]}p$ then $\phi_L=\Diamond(p\wedge q)\wedge\Diamond p$.
In essence, if $\phi_M$ is satisfiable, then $\phi_L$ is also satisfiable. Therefore, if $\phi_L$ is unsatisfiable, then $\phi_M$ is also unsatisfiable.

For the always ($\Box$) operator, satisfiability is the dual of the eventually operator ($\Diamond$). Assume that $\phi_M'\in$ MITL($\Box$) contains only the $\Box$ operator and $\phi_L'\in$ LTL($\Box$) is the LTL version of $\phi_M'$. If $\phi_L'$ is satisfiable, then $\phi_M'$ will also be satisfiable.

Based on the above discussion, if the specification that we intend to test/debug belongs to either category (fragment),  MITL($\Diamond$) or MITL($\Box$), then we can check the satisfiability of its LTL version ($\phi_L$) and decide according to the following:
\begin{theorem}
	\label{thm:MITL2LTL}
	For any formula $\phi_M\in$ MITL($\Diamond$) and $\phi_M'\in$ MITL($\Box$) then
	
	If $\phi_L\in$ LTL($\Diamond$) is unsatisfiable, then $\phi_M$ is unsatisfiable.

	If $\phi_L'\in$ LTL($\Box$) is satisfiable, then $\phi_M'$ is satisfiable.
\end{theorem}

In these two cases, we do not need to run MITL SAT, if otherwise, we must apply MITL SAT which means that we wasted effort by checking LTL SAT. However, since the runtime of LTL SAT is negligible, it will not drastically reduce the performance. As a result LTL satisfiability checking is useful for validity testing. For redundancy checks, it may also be useful. For example, if we have a formula $\phi=\Diamond_{[0,10]} p\wedge\Box_{[0,20]} p$ we should check the satisfiability of $\phi'=\Box_{[0,10]} \neg p\wedge\Box_{[0,20]} p$ and $\phi''=\Diamond_{[0,10]} p\wedge\Diamond_{[0,20]} \neg p$ for redundancy. 
Although the original formula $\phi$  does not belong to either MITL($\Diamond$) or MITL($\Box$), its modified NNF version will fit in these fragments and we may benefit by the usually faster LTL satisfiability for $\phi'$ and/or $\phi''$. 
For vacuity checking, we can use LTL satisfiability if after manipulating/simplifying the original specification and creating the NNF version, we can categorize the resulting formula into the MITL($\Diamond$) or the MITL($\Box$) fragments (see Table  \ref{tab:ltlResults}). 

We can check LTL satisfiability of the modified MITL specifications using existing methods and tools \cite{RozierLTLSAT10}. In our case, we used the NuSMV\footnote{ NuSMV Version 2.6.0. Available from \url{http://nusmv.fbk.eu/}} tool with a similar encoding of LTL formulas as in \cite{RozierLTLSAT10}.
In Table \ref{tab:ltlResults}, we compare the runtime overhead of MITL and LTL satisfiability checking. For the results of the usability study in \cite{Hoxha_ViSpecIROS15}, we conduct validity and vacuity checking with the LTL satisfiability solver. We remark that in our results in Table \ref{tab:ltlResults}, all the formulas belong to the MITL($\Box$) fragment. Since we did not find any MITL($\Diamond$) formula in our experiments where its LTL version is not satisfiable, we did not provide MITL($\Diamond$) formulas in Table \ref{tab:ltlResults}.

The first column of Table \ref{tab:ltlResults} provides the debugging test phase where we used the satisfiability checkers. The second column represents the MITL formulas that we tested using the SAT solver. We omit the LTL formulas from Table \ref{tab:ltlResults}, since they are identical to MITL but do not contain timing intervals. The atomic propositions  $p_1,p_2,p_3,p_4,p_5$ of the MITL formulas in Table \ref{tab:ltlResults} correspond to the following STL predicates: $p_1 \equiv speed>100$, $p_2 \equiv rpm>4000$, $p_3 \equiv 100\ge speed>80$, $p_4 \equiv rpm>3000$, and $p_5 \equiv speed>80$. The third and fourth columns represent the runtime overhead of satisfiability checking for MITL specifications and their corresponding LTL version. The last column represents the speedup of the LTL approach over the MITL approach. It can be seen that the LTL SAT solver (NuSMV) is about 30-300 times faster than the MITL SAT solver (zot). These results confirm that, when applicable, LTL SAT solvers outperform MITL SAT solvers in checking vacuity and validity issues in specifications. 
As a result, \yhl{it is worth running} LTL SAT before MITL SAT when it is possible.
\begin{table}[t]
	\centering
	\ifthenelse{\boolean{ARXIV}}{\caption{Comparing the runtime overhead of MITL satisfiability and LTL satisfiability (in Seconds) for some of the specifications from \textsc{ViSpec}'s usability study. \label{tab:ltlResults}}}{\tbl{Comparing the runtime overhead of MITL satisfiability and LTL satisfiability (in Seconds) for some of the specifications from \textsc{ViSpec}'s usability study. \label{tab:ltlResults}}}
	{
		\begin{tabular}{| c | c | c | c |c |}
			\hline
			Test & MITL Specification  & MITL  &  LTL  & MITL / LTL  \\
			Phase&   &   SAT & SAT & runtime  \\\hline\hline
			Validity & $\Box_{[0,40]}(p_1\Rightarrow\Box_{[0,10]}(p_1))$ & 4.154 & 0.047  & 88  \\ \hline 
			Validity & $\Box_{[0,30]}(\neg p_1)\vee\Box_{[0,20]}(\neg p_1)$ & 3.418  & 0.0538  & 63  \\ \hline 
			Validity & $\Box_{[0,40]}((\neg p_1\wedge \neg p_3)\vee\Box_{[0,20]}\neg p_2\vee\Box_{[0,30]}p_1))$ & 10.85  & 0.045  & 240  \\ \hline
			Validity & $\Box_{[0,40]}((p_1\vee p_3)\Rightarrow\Box_{[0,20]}(p_2\Rightarrow\Box_{[0,30]}p_1))$ & 15.406  & 0.0463  & 333  \\ \hline
			Vacuity &$\Box_{[0,40]}(p_1)$ & 1.71 & 0.0473 & 36 \\ \hline
			Vacuity & $\Box_{[0,40]}(p_1\wedge\Box_{[0,10]}(p_1))$ & 3.727 & 0.044  & 84  \\ \hline 
			Vacuity &$\Box_{[0,40]}p_1\wedge\Box_{[0,30]}(p_4)$ & 5.77 &  0.0456  & 126 \\ \hline
			Vacuity & $\Box_{[0,40]}p_5\wedge\Box_{[0,70]}(p_5)$  & 8.599 & 0.044 & 194 \\ \hline

		\end{tabular}}
	\end{table}
	
\subsection{Antecedent Failure Detection}
\label{afd}
	To apply signal vacuity checking we use the \staliro testing framework \cite{AbbasTECS2013,hoxhatowards}.
\staliro is a MATLAB toolbox that uses stochastic optimization techniques to search for system inputs for Simulink models which falsify the safety requirements presented in MTL/STL \cite{AbbasTECS2013}.
 The signal vacuity checking implemented in the \staliro tool is computationally efficient (PTIME). Its time complexity is proportional to the number of implication operations, the size of the formula and to the size of the signal \cite{Fainekos2012}.

In the following, we illustrate the vacuous signal detection process by using the Automatic Transmission (AT) model provided by Mathworks as a Simulink demo\footnote{Available at: http://www.mathworks.com/help/simulink/examples/modeling-an-automatic-transmission-controller.html}.
We introduced a few modifications to the model to make it compatible with the \staliro framework. Further details can be found in \cite{HoxhaAF14arch}. 
\staliro calls the AT Simulink model in order to generate the output trajectories.
The outputs contain two continuous-time real-valued signals: the speed of the engine $\omega$ (RPM) and the speed of the vehicle 
$v$. 
In addition, the outputs contain one continuous-time discrete-valued signal $gear$ with four possible values 
($gear=1$, ..., $gear=4$) which indicates the current gear in the auto-transmission controller. \staliro 
then monitors system trajectories with respect to the requirements provided in Table \ref{tab:ATreqs}.
There, in the MITL formulas, we use the shorthand $g_i$ to indicate the gear value, i.e. $(gear=i)\equiv g_i$. The simulation time for the system is set to 30 seconds; therefore, we can use bounded MITL formulas for the requirements. 

After testing the AT with \staliro, we collected all the system trajectories. 
Then, we utilized the antecedent failure mutation on the specification to check signal vacuity (Algorithm \ref{alg:antec}) for each of the formulas that are provided in Table \ref{tab:ATreqs}. We provide the antecedent failure specifications and \yhl{the number of signals that satisfy them} in Table \ref{tab:AFreqs}. 
It can be seen in Table \ref{tab:AFreqs} that most of the system traces are vacuous signals where the antecedent is not satisfied. This helps the users to consider these issues and identify interesting test cases that can be used to initialize the system tester so that the antecedent is always satisfied. \yhl{For an application of signal vacuity checking in falsification see} \cite{DokhanchiYHF17}.

\begin{table}[t]
	\centering
	\ifthenelse{\boolean{ARXIV}}{\caption{Automatic Transmission Requirements expressed in natural language and MITL from  \cite{HoxhaAF14arch}\label{tab:ATreqs}}}{\tbl{Automatic Transmission Requirements expressed in natural language and MITL from  \cite{HoxhaAF14arch}\label{tab:ATreqs}}}
	{
		\begin{tabular}{| c | l | c |}
			\hline
			Req. & {\bf Natural Language}  & MITL Formula \\ \hline 
			
			\multirow{2}{*}{$\phi^{AT}_{1}$}    & There should be no transition from gear two to gear & \multirow{2}{*}{$\Box_{[0,27.5]}((g_2\wedge\Diamond_{(0,0.04]}g_1)\Rightarrow\Box_{[0,2.5]}\neg g_2)$} \\
			&  one and back to gear two in less than 2.5 sec. &  \\  \hline 
			\multirow{2}{*}{$\phi^{AT}_{2}$}                & After shifting into gear one, there should be no shift & \multirow{2}{*}{$\Box_{[0,27.5]}((\neg g_1\wedge\Diamond_{(0,0.04]}g_1)\Rightarrow\Box_{[0,2.5]} g_1)$} \\
			&  from gear one to any other gear within 2.5 sec. &  \\\hline 
			\multirow{2}{*}{$\phi^{AT}_{3}$}                & If the $\omega$ is always less than 4500, then the $v$ can not  & \multirow{2}{*}{$\Box_{[0,30]}(\omega\le 4500) \Rightarrow\Box_{[0,10]}(v\le 85)$} \\
			& exceed 85 in less than 10 sec. &  \\  \hline 
			\multirow{2}{*}{$\phi^{AT}_{4}$} & Within 10 sec. $v$ is more than 80 and from that point& \multirow{2}{*}{$\Diamond_{[0,10]}((v\ge 80) \Rightarrow\Box_{[0,30]}(\omega\le 4500))$} \\
			& on, $\omega$ is always less than 4500. &   \\  \hline 
			
	\end{tabular}}
\end{table}
\begin{table}[t]
	\centering
	\ifthenelse{\boolean{ARXIV}}{\caption{Reporting signal vacuity issue for each mutated formula \label{tab:AFreqs}}}{\tbl{Reporting signal vacuity issue for each mutated formula \label{tab:AFreqs}}}	
	{
		\begin{tabular}{| c | l | c |}
			\hline
			Requirement	& {\bf Antecedent Failure Mutation}  & Vacuous Signals / All Signals \\ \hline 
			
			$\phi^{AT}_{1}$ & $\Box_{[0,27.5]}\neg(g_2\wedge\Diamond_{(0,0.04]}g_1)$ & 1989 / 2000 \\
			\hline 
			$\phi^{AT}_{2}$ & $\Box_{[0,27.5]}\neg(\neg g_1\wedge\Diamond_{(0,0.04]}g_1)$  &  1994 / 2000 \\ \hline 
			$\phi^{AT}_{3}$ & $\neg\Box_{[0,30]}(\omega\le 4500) $  & 97 / 307 \\  \hline 
			$\phi^{AT}_{4}$ &  $\Box_{[0,10]}\neg(v\ge 80)$ & 1996 / 2000 \\  \hline 
			
	\end{tabular}}
\end{table}
\section{Conclusion and Future Work}

We have presented a specification elicitation and debugging framework that can assist \textsc{ViSpec} users to produce correct formal specifications. In particular, the debugging algorithm enables the detection of logical inconsistencies in MITL and STL specifications. Our algorithm improves the elicitation process by providing feedback to the users on validity, redundancy and vacuity issues. 
In the future, the specification elicitation and debugging framework will be integrated in the \textsc{ViSpec} tool to simplify MITL and STL specification development for verification of CPS.
In addition, we considered vacuity detection with respect to signals. 
This enables improved analysis since some issues can only be detected when considering both the system and the specification.
In the future, we will consider the feasibility of using vacuous signals to improve the counter example generation process and system debugging using signal vacuity.

\section*{Acknowledgment}

This work was partially supported by NSF awards CNS 1350420 and CNS 1319560.

\bibliographystyle{abbrv}
\bibliography{bardh,fatesrv,adel}
\section*{APPENDIX}
\section{ Proof of Theorem \ref{thm:vacu}}
\label{app:Thm1}

In order to show that $\Phi$ is inherently vacuous, we must show that if $\Phi\models\varphi_i[l\leftarrow\perp]$, then the mutated specification is equivalent to the original specification. In other words, we should show that if $\Phi\models\varphi_i[l\leftarrow\perp]$, then $(\{\Phi\backslash\varphi_i\}\cup\varphi_i[l\leftarrow\perp])\equiv\Phi$ . If the mutated specification is equivalent to the original specification, then the original specification is vacuously satisfiable in any system. That is, the specification is {\it inherently vacuous} \cite{FismanKSV08,ChocklerS09}. 
 We already know that if $\Phi\models\varphi_i[l\leftarrow\perp]$, then $\Phi\implies\varphi_i[l\leftarrow\perp]$ and trivially $\Phi\implies\varphi_i[l\leftarrow\perp]\cup\{\Phi\backslash\varphi_i\}$. Now we just need to prove the other direction. 
We need to prove that when $\varphi_i$ \yhl{is in NNF,} then $\varphi_i[l\leftarrow\perp]\implies\varphi_i$. 
Since we replace only one specific literal occurrence of $\varphi$ with $\bot$, the rest of the formula remains the same. Therefore, it should be noted that $\varphi_i[l\leftarrow\perp]$ does not modify any $l'\in litOccur(\varphi_i)$ where $l'\ne l$.
\begin{proof}
We use structural induction to prove that $\varphi_i[l\leftarrow\perp]\implies\varphi_i$\\
{\bf Base Case:} $\varphi_i=l$ or $\varphi_i=l'\ne l$\\
We know that $\bot\implies l$ and $l'\implies l'$. Therefore $\varphi_i[l\leftarrow\perp]\implies\varphi_i$.\\
{\bf Induction Hypothesis:} For any MITL $\varphi_j$ in NNF we have $\varphi_j[l\leftarrow\perp]\implies\varphi_j$ (or $\forall\varphi_j,\varphi_j[l\leftarrow\perp]\implies\varphi_j$) \\
{\bf Induction Step:} We will separate the case into unary and binary operators.\\Before providing the cases we should review the {\it positively monotonic} operators \cite{KupfermanV03}. According to MITL semantics, $f\in\{\Box_I,\Diamond_I\}$ and $g\in\{\wedge,\vee\}$ are positively monotonic, i.e. for every MITL formulas $\varphi_1$ and $\varphi_2$ in NNF with $\varphi_1\implies\varphi_2$, we have $f(\varphi_1)\implies f(\varphi_2)$. Also, for all MITL formulas $\varphi'$ in NNF, we have $g(\varphi_1,\varphi')\implies g(\varphi_2,\varphi')$ and $g(\varphi',\varphi_1)\implies g(\varphi',\varphi_2)$.
\\{\bf Case 1:}  $\varphi_i=f(\varphi_j)$ where $f\in\{\Box_I,\Diamond_I\}$.
 Since $f$ is positively monotonic, we have that $\varphi_j[l\leftarrow\perp]\implies\varphi_j$ implies  $f(\varphi_j[l\leftarrow\perp])\implies f(\varphi_j)$. Thus,\\ $f(\varphi_j)[l\leftarrow\perp]=f(\varphi_j[l\leftarrow\perp])\implies f(\varphi_j)=\varphi_i$. As a result $\varphi_i[l\leftarrow\perp]\implies\varphi_i$.
\\{\bf Case 2:}  $\varphi_i=g(\varphi_{j_1},\varphi_{j_2})$ where $g\in\{\wedge,\vee\}$
 Since $g$ is positively monotonic, we have that $\varphi_{j_1}[l\leftarrow\perp]\implies\varphi_{j_1}$, and $\varphi_{j_2}[l\leftarrow\perp]\implies\varphi_{j2}$ implies\\ $g(\varphi_{j_1}[l\leftarrow\perp],\varphi_{j_2}[l\leftarrow\perp])\implies g(\varphi_{j_1},\varphi_{j_2})$ . Thus,  $g(\varphi_{j_1},\varphi_{j_2})[l\leftarrow\perp]=g(\varphi_{j_1}[l\leftarrow\perp],\varphi_{j_2}[l\leftarrow\perp])\implies g(\varphi_{j_1},\varphi_{j_2})=\varphi_i$. As a result $\varphi_i[l\leftarrow\perp]\implies\varphi_i$. 

Since $\varphi_i[l\leftarrow\perp]\implies\varphi_i$ we can have:\\ $\{\Phi\backslash\varphi_i\}\cup\varphi_i[l\leftarrow\perp]\implies\{\Phi\backslash\varphi_i\}\cup\varphi_i$ which is equivalent to \\ $\{\Phi\backslash\varphi_i\}\cup\varphi_i[l\leftarrow\perp]\implies\Phi$ 
\end{proof}

\section{ Proof of Theorem \ref{thm:sigVacuity}}
\label{app:SV}
In this section, we will prove that any MITL (STL) $\varphi\in\Phi$, which contains a disjunction operation ($\vee$) in NNF can be satisfied \yhl{by partially covering signals.} In other words, we will prove that any timed trace (signal) which satisfies $\varphi$ will be considered as a \yhl{partially covering} timed trace (signal) according to Algorithm \ref{alg:litOcRe}. Without loss of generality we assume that both operands of disjunction are not constant. This is because if one of the operands is equivalent to $\top$ or $\perp$, then the disjunction can be semantically removed as follows $\psi\vee\top\equiv\top$ or $\psi\vee\perp\equiv\psi$ for any MITL (STL) $\psi$.

Let us consider a the \yhl{partially covering} timed trace (signal) returned by Algorithm \ref{alg:litOcRe}. 
If there exist $\varphi_i\in\Phi$ and $l\in litOccur(\varphi_i)$ such that the timed trace $\tss$ satisfies $\varphi_i[l\leftarrow\perp]$, then $\tss$ will be reported as a \yhl{partially covering} timed trace (signal). Recall that we assume that $\Phi$ is a conjunction of MITL specifications according to Equation (\ref{eq:conjunction}). We also assume that the conjunct $\varphi_i\in\Phi$ is the MITL subformula that contains the disjunction operation. Namely, that $\psi=\psi_1\vee\psi_2$ is a subformula of $\varphi_i$. 

\begin{theorem}
Any timed trace $\tss$ that satisfies $\varphi_i$ will satisfy $\varphi_i[l\leftarrow\perp]$ for some $l\in litOccur(\varphi_i)$.
\end{theorem}
\begin{proof}
We have two cases for $\tss\models\varphi_i$ and $\psi\in\varphi_i$ where $\psi=\psi_1\vee\psi_2$:
\begin{enumerate}
	\item $\forall t,(\tss,t)\not\models\psi$: In this case $\psi$ does not affect the satisfaction of $\tss\models\varphi_i$. If we choose $l'\in litOccur(\psi)$, then $\psi[l'\leftarrow\perp]$ also does not affect the satisfaction of $\varphi_i$.
	\item $\exists t,\mbox{ s.t. }(\tss,t)\models\psi$: In this case $\psi$ affects the satisfaction of $\varphi_i$. 
	So either $(\tss,t)\models\psi_1$ or $(\tss,t)\models\psi_2$. 
	If $(\tss,t)\models\psi_1$ then we can choose $l'\in litOccur(\psi_2)$ and we have $(\tss,t)\models\psi_1\vee\psi_2[l'\leftarrow\perp]$. 
	Similarly, if $(\tss,t)\models\psi_2$ then we can choose $l''\in litOccur(\psi_1)$ and we have $(\tss,t)\models\psi_1[l''\leftarrow\perp]\vee\psi_2$. As a result, there exists some $l\in litOccur(\psi)$ where $(\tss,t)\models\psi[l\leftarrow\perp]$ and accordingly $\tss\models\varphi_i[l\leftarrow\perp]$.
\end{enumerate}

Finally, $\forall\tss, \tss\models\varphi_i\mbox{ }\exists l\in litOccur(\varphi_i)\mbox{ s.t. }\tss\models\varphi_i[l\leftarrow\perp]$. Which means that $\tss$ is a \yhl{partially covering} timed trace (signal).\end{proof}

\begin{corollary}
	\label{cor:con}
	Assume that the conjunct $\varphi_j\in\Phi$ is the subformula that contains the conjunction operation in NNF. Namely, that $\psi=\psi_1\wedge\psi_2$ is a subformula of $\varphi_j$. Any timed trace $\tss$ that falsifies $\varphi_j$ will falsify $\varphi_j[l\leftarrow\top]$ for some $l\in litOccur(\varphi_j)$.
\end{corollary}

\section{Proofs of Theorem \ref{thm:MITL2LTL}}
\label{app:MITLvsLTL}
We consider two MITL($\Diamond$,$\Box$) fragments, denoted MITL($\Box$), and MITL($\Diamond$). In this proof we assume that all formulas are in NNF. We also consider LTL($\Diamond$,$\Box$) as the set of LTL formulas (with continuous semantics) that contains only $\Diamond$ and $\Box$ as temporal operators. In the following we provide the continuous semantics of LTL($\Diamond$,$\Box$) over traces with bounded duration. Semantics of LTL($\Diamond$,$\Box$) over bounded timed traces can be defined as follows:
\begin{definition}[LTL($\Diamond$,$\Box$) continuous semantics] Given a timed trace $\tss : [0,T] \rightarrow 2^{AP}$  and $t,t' \in \Re$, and an LTL($\Diamond$,$\Box$) formula $\phi$, the satisfaction relation $(\tss,t) \vDash \phi$ for temporal operators is inductively defined: 
	\label{def:ltl}
	\begin{itemize}
		\item[] $(\tss,t) \vDash  \Diamond\phi_1$ iff $\exists t'\in [t,T]$ s.t $(\tss,t') \vDash  \phi_1$.
		\item[] $(\tss,t) \vDash  \Box\phi_1$ iff $\forall t'\in [t,T]$,  $(\tss,t') \vDash \phi_1$.
	\end{itemize}
	
\end{definition} 
We will consider two LTL($\Diamond$,$\Box$) fragments denoted LTL($\Box$), and LTL($\Diamond$). The syntax of MITL and LTL fragments are as presented in \yhl{Section }\ref{ltlsat}.
We define the operator $[\phi]_{LTL}$ which can be applied to any MITL($\Diamond$,$\Box$) formula and removes its interval constraints to create a new formula in LTL($\Diamond$,$\Box$). For example if $\phi=\Diamond_{[0,10]}(p\wedge q)\wedge\Diamond_{[0,10]} p\wedge\Box_{[0,10]} q$, then $[\phi]_{LTL}=\Diamond(p\wedge q)\wedge\Diamond p\wedge\Box q$. As a result, for any $\phi\in$ MITL$(\Diamond,\Box)$ there exists a $\psi\in$ LTL$(\Diamond,\Box)$ where $\psi=[\phi]_{LTL}$. For each MITL$(\Diamond,\Box)$ formula $\phi$, the language of $\phi$ denoted $L(\phi)$ is the set of all timed traces that satisfy $\phi$:
$\tss \vDash \phi$ iff $\tss\in L(\phi)$. Similarly, for any $\psi\in$ LTL$(\Diamond,\Box)$, the language of $\psi$ denoted $L(\psi)$ is the set of all timed traces that satisfy $\psi$:
$\tss' \vDash \psi$ iff $\tss'\in L(\psi)$. Based on set theory, it is trivial to prove that $A\subseteq B$ and $C\subseteq D$ implies  $A\cup C \subseteq B\cup D$ and $A\cap C \subseteq B\cap D$.
\begin{theorem}
	\label{thm:diam}
	For any formula $\varphi\in$ MITL($\Diamond$), and $t\in[0,T]$ we have $L_t(\varphi)\subseteq L_t([\varphi]_{LTL})$ where $L_t(\varphi)=\{\tss\; | \;(\tss,t)\vDash\varphi\}$. In other words\yhl{, for every timed trace $\tss$,} we have $(\tss,t)\vDash\varphi$ implies $(\tss,t)\vDash[\varphi]_{LTL}$.
\end{theorem}
\begin{proof}
	We use structural induction to prove that $L_t(\varphi)\subseteq L_t([\varphi]_{LTL})$\\
	{\bf Base Case:} if $\varphi=\top,\bot, p, \neg p$, then $[\varphi]_{LTL}=\varphi$ and $L_t(\varphi)\subseteq L_t([\varphi]_{LTL})$\\
	{\bf Induction Hypothesis:} We assume that there exist $\varphi_1,\varphi_2\in$ MITL($\Diamond$) where for all $t\in[0,T]$, $L_t(\varphi_1)\subseteq L_t([\varphi_1]_{LTL})$ and $L_t(\varphi_2)\subseteq L_t([\varphi_2]_{LTL})$
	
	{\bf Case 1:} For Binary operators $\wedge,\vee$ we can use the union and intersection properties. In essence, for all formulas $\varphi_1,\varphi_2$ we have $L_t(\varphi_1\vee\varphi_2)=L_t(\varphi_1)\cup L_t(\varphi_2)$ and $L_t(\varphi_1\wedge\varphi_2)=L_t(\varphi_1)\cap L_t(\varphi_2)$. According to the IH $L_t(\varphi_1)\subseteq L_t([\varphi_1]_{LTL})$ and $L_t(\varphi_2)\subseteq L_t([\varphi_2]_{LTL})$; therefore, $L_t(\varphi_1)\cap L_t(\varphi_2)\subseteq L_t([\varphi_1]_{LTL})\cap L_t([\varphi_2]_{LTL})$ and $L_t(\varphi_1)\cup L_t(\varphi_2)\subseteq L_t([\varphi_1]_{LTL})\cup L_t([\varphi_2]_{LTL})$. As a result, $L_t(\varphi_1\wedge\varphi_2)\subseteq L_t([\varphi_1]_{LTL}\wedge[\varphi_2]_{LTL})=L_t([\varphi_1\wedge\varphi_2]_{LTL})$, and $L_t(\varphi_1\vee\varphi_2)\subseteq L_t([\varphi_1]_{LTL}\vee[\varphi_2]_{LTL})=L_t([\varphi_1\vee\varphi_2]_{LTL})$.
	
	{\bf Case 2:} For the temporal operator $\Diamond$, we need to compare the semantics of MITL($\Diamond$) and LTL($\Diamond$). Recall that
	
	$(\tss,t) \vDash  \Diamond_I\varphi_1$ iff $\exists t' \in (t + I) \cap [0,T]$ s.t $(\tss,t') \vDash  \varphi_1$.
	
	$(\tss,t) \vDash  \Diamond\varphi_1$ iff $\exists t'\in [t,T]$ s.t $(\tss,t') \vDash  \varphi_1$.\\
	Recall that $t'' \in (t + I) \cap [0,T]$ implies $t'' \in [t,T]$ since the left bound of $I$ is nonnegative.
	
	According to the semantics, $\forall\tss$.$(\tss,t) \vDash  \Diamond_I\varphi_1$ implies \\$ \exists t' \in (t + I) \cap [0,T]$ s.t $(\tss,t') \vDash  \varphi_1$ implies \\$\exists t' \in (t + I) \cap [0,T]$ s.t $ (\tss,t') \vDash  [\varphi_1]_{LTL}$ according to IH ($L_{t'}(\varphi_1)\subseteq L_{t'}([\varphi_1]_{LTL})$).\\
	If $\exists t' \in (t + I) \cap [0,T]$ s.t $ (\tss,t') \vDash  [\varphi_1]_{LTL}$ then \\
	$\exists t' \in [t,T]$ s.t $ (\tss,t') \vDash  [\varphi_1]_{LTL}$ since  $t' \in (t + I) \cap [0,T]$ implies $t' \in [t,T]$.
	
	Moreover, $(\tss,t') \vDash  [\varphi_1]_{LTL}$ implies that $ (\tss,t) \vDash  \Diamond[\varphi_1]_{LTL}\equiv[\Diamond\varphi_1]_{LTL}$.\\ As a result, 
	$\forall \tss$. $(\tss,t) \vDash  \Diamond_I\varphi_1 \implies (\tss,t) \vDash[\Diamond\varphi_1]_{LTL}$ so $L_t(\Diamond_I\varphi_1)\subseteq L_t([\Diamond\varphi_1]_{LTL})$.
\end{proof}
If $\varphi\in$ MITL($\Diamond$) then $\overline{L_t([\varphi]_{LTL})}\subseteq\overline{L_t(\varphi)}$ (immediate from set theory). Thus, for all timed traces $\tss$, $\tss\not\vDash[\varphi]_{LTL}$ implies that $\tss\not\vDash\varphi$.
\begin{corollary}
	For any $\varphi\in$ MITL($\Diamond$), if $[\varphi]_{LTL}\in$ LTL($\Diamond$) is unsatisfiable, then $\varphi$ is unsatisfiable.
\end{corollary}
\begin{theorem}
	\label{thm:box}
	For any formula $\varphi\in$ MITL($\Box$), and $t\in[0,T]$, we have $L_t([\varphi]_{LTL})\subseteq L_t(\varphi)$, where $L_t(\varphi)=\{\tss|(\tss,t)\vDash\varphi\}$. \yhl{In other words, $\forall \tss (\tss,t)\vDash[\varphi]_{LTL}\implies(\tss,t)\vDash\varphi$.}
\end{theorem}
\begin{proof}
	Similar to Theorem \ref{thm:diam}, we can apply structural induction for the proof of Theorem \ref{thm:box}.\end{proof}

\end{document}